\documentclass[a4paper,UKenglish,cleveref,autoref,thm-restate]{arxiv}

\nolinenumbers

\usepackage{stmaryrd}
\usepackage{tikz-cd}
\hypersetup{
    colorlinks = true,
    linkcolor={red},citecolor={red},urlcolor={magenta}
}

\definecolor{cornflowerblue}{rgb}{0.99,0.78,0.07}
\definecolor[named]{lipicsYellow}{rgb}{0.99,0.78,0.07}

\usepackage{mathtools}
\newcommand{\defeq}{\vcentcolon=}

\newcommand{\ket}[1]{  |{#1} \rangle} 

\renewcommand{\iff}{\Leftrightarrow}
\newcommand{\se}{\subseteq}
\newcommand{\ls}{\leqslant}
\newcommand{\gs}{\geqslant}
\newcommand{\sm}{\setminus}

\newcommand{\Zp}[1]{\mathbb{Z}/#1\mathbb{Z}}

\DeclareSymbolFont{matha}{OML}{txmi}{m}{it}

\definecolor{fillvertices}{RGB}{250,240,230}

\newcommand{\supp}{\textup{\textsf{supp}}}

\renewcommand{\iff}{\Leftrightarrow}
\renewcommand{\iff}{\Leftrightarrow}

\allowdisplaybreaks

\AddToHook{cmd/appendix/before}{\crefalias{section}{appendix}}

\title{Deciding  Local Unitary Equivalence of Graph States in Quasi-Polynomial Time}

\author{Nathan Claudet}{Inria Mocqua and Université de Lorraine, CNRS, LORIA, F-54000 Nancy, France
}{nathan.claudet@inria.fr}{https://orcid.org/0009-0000-0862-1264}{}

\author{Simon Perdrix}{Inria Mocqua and Université de Lorraine, CNRS, LORIA, F-54000 Nancy, France}{simon.perdrix@loria.fr}{https://orcid.org/0000-0002-1808-2409}{}

\authorrunning{N. Claudet and S. Perdrix}

\Copyright{Nathan Claudet and Simon Perdrix}


\ArticleNo{23} 

\begin{document}

\maketitle

\begin{abstract}  

    We describe an algorithm with quasi-polynomial runtime $n^{\log_2(n)+O(1)}$ for deciding local unitary (LU) equivalence of graph states. The algorithm builds on a recent graphical characterisation of LU-equivalence via generalised local complementation. By first transforming the corresponding graphs into a standard form using usual local complementations, LU-equivalence reduces to the existence of  a single generalised local complementation that maps one graph to the other. We crucially demonstrate that this reduces to solving a system of quasi-polynomially many linear equations, avoiding an exponential blow-up. As a byproduct, we generalise Bouchet's algorithm for deciding local Clifford (LC) equivalence of graph states by allowing the addition of arbitrary linear constraints. We also improve existing bounds on the size of graph states that are LU- but not LC-equivalent. While the smallest known examples involve 27 qubits, and it is established that no such examples exist for up to 8 qubits, we refine this bound by proving that LU- and LC-equivalence coincide for graph states involving up to 19 qubits.
\end{abstract}

\section{Introduction}

Graph states form a ubiquitous family of quantum states. They are used as entangled resource states in various quantum information applications, such as measurement-based computation \cite{raussendorf2001one, raussendorf2003measurement, briegel2009measurement}, error correction \cite{schlingemann2001quantum,schlingemann2001stabilizer,cross2008codeword,sarvepalli2011local}, quantum communication network routing \cite{hahn2019quantum,meignant2019distributing, bravyi2024generating, Cautres2024}, and quantum secret sharing \cite{markham2008graph,gravier2013quantum},  to cite a few.  
In all these applications, graph states are used as multipartite entangled resources, it is thus crucial to understand when two such states have the same entanglement, i.e.~when they can be transformed into each other using only local operations. SLOCC-equivalence (stochastic local operations and classical communications) is the most general case that encompasses the use of local unitaries and measurements. In the particular case of graph states, it is enough to consider \emph{LU-equivalence} (local unitaries), as two graph states are SLOCC-equivalent if and only if there exists $U=U_1\otimes \ldots \otimes U_n$ that transforms one state into the other, where each $U_i$ is a single-qubit unitary transformation \cite{Hein04}. One can also consider \emph{LC-equivalence} (local Clifford) which is known to be distinct from LU-equivalence, the smallest known examples of graph states that are LU-equivalent but not LC-equivalent have 27 qubits \cite{Ji07,Tsimakuridze17}. 

As their name suggests, graph states can be uniquely represented by simple undirected graphs. Remarkably, LC-equivalence of graph states is captured by applications of a simple transformation on the corresponding graphs: \emph{local complementation} \cite{VandenNest04}. Local complementation consists in complementing the neighbourhood of a given vertex. Local complementation was introduced by Kotzig in the 1960s \cite{Kotzig68}, and has been studied independently of its applications in quantum computing. In particular, Bouchet has introduced an efficient algorithm for deciding whether two graphs are related by a sequence of local complementations \cite{Bouchet1991}. This has led to an efficient algorithm for deciding the local Clifford equivalence of graph states within $O(n^4)$ operations, where $n$ is the number of qubits \cite{VdnEfficientLC}.

Recently a graphical characterisation of LU-equivalence has been introduced by means of \emph{generalised local complementation} \cite{claudet2024local}. The characterisation relies in particular on some peculiar graph structures called \emph{minimal local sets} that are known to be invariant under local unitary transformations \cite{Perdrix06}. In \cite{claudet2024covering}, it was shown that any vertex is covered by a minimal local set and that a family of minimal local sets covering every vertex of the graph, called an \emph{MLS cover}, can be computed efficiently. Roughly speaking each minimal local set imposes a constraint on the local unitary transformations mapping a graph state to another, so that the existence of such a local unitary is reduced to solving a linear system  over integers modulo a power of 2. The solutions can then be graphically interpreted as generalised local complementations. 

Shortly after, an algorithm for deciding LU-equivalence was independently  introduced \cite{burchardt2024algorithmverifylocalequivalence} based on a similar idea of reducing the problem of LU-equivalence to a linear system, benefiting in particular from the fact that an MLS cover can be computed efficiently. The overall complexity of this algorithm for deciding LU-equivalence depends on two parameters, roughly speaking the size of the linear system and the number of connected components of an intersection graph related to the MLS cover. Both parameters can potentially make the runtime of the algorithm exponential. 

We introduce a new algorithm for LU-equivalence of graph states that relies on generalised local complementation and allows us to 
mitigate both sources of exponential complexity. 
First, we reduce the LU-equivalence problem to the existence of a single generalised local complementation. To achieve this efficient reduction, we extend Bouchet's algorithm. Then, we demonstrate that the level of the remaining generalised local complementation can be upper bounded by at most the logarithm of the order $n$ of the graphs, leading to a linear system of size at most $n^{\log_2(n)+O(1)}$. This results in an overall algorithm whose time-complexity is quasi-polynomial in $n$. Notice that the generalisation of Bouchet's algorithm provides an efficient algorithm for deciding whether two graphs are related by local complementations under additional constraints, for instance that the local complementations are applied to a particular subset of vertices.

Thanks to the graphical characterisation of LU-equivalence by means of generalised local complementation, we also address the question of the smallest graphs that are LU- but not LC-equivalent. The study of graph classes where LU-equivalence coincides with LC-equivalence has garnered significant attention  
\cite{Hein04, Hein06, VandenNest05,Zeng07,Ji07,CABELLO20092219,tzitrin2018local,claudet2024local,burchardt2024algorithmverifylocalequivalence}. Notably, the smallest known examples of graphs that are LU- but not LC-equivalent have 27 vertices while it is established that no such counterexamples exist for fewer than 8 vertices \cite{CABELLO20092219}. We significantly  improve this result by showing that any counterexample has at least 20 vertices.  

\section{Preliminaries}

{\bf Notations.} Given an undirected simple graph $G=(V,E)$, we use the notation $u\sim_G v$ when the vertices $u$ and $v$ are connected in $G$, i.e.~$(u,v)\in E$. $N_G(u)=\{v\in V~|~u\sim_G v\}$ is the neighbourhood of $u$, $Odd_G(D) = \{v \in V~|~|N_G(v)\cap D|=1\bmod 2\}$ is the odd-neighbourhood of the set $D\subseteq V$ of vertices, and $\Lambda _G^D = \{v\in V~|~ \forall u\in D, u\sim_G v\}$ is the common neighbourhood of $D\subseteq V$. We assume $V$  totally ordered by a relation $\prec$. A local complementation with respect to a given vertex $u$ consists in complementing the subgraph induced by the neighbourhood of $u$, leading to the graph $G\star u= G\Delta K_{N_G(u)}$ where $\Delta$ denotes the symmetric difference on edges and $K_A$ is the complete graph on the vertices of $A$. With a slight abuse of notation we identify multisets of vertices with their multiplicity function $V\to \mathbb N$ (hence we also identify sets of vertices with their indicator functions $V\to \{0,1\}$). We consider sums of multisets: for any vertex $u$, $(S_1+ S_2)(u) = S_1(u)+S_2(u)$. The support $\supp(S)$ of a multiset $S$ of vertices denotes the set of vertices $u \in V$ such that $S(u) > 0$. $S$ is said independent if no two vertices of $\supp(S)$ are connected.
For any multiset $S$ and set $D$, $S\bullet \Lambda_G^D$ is the number of vertices of $S$, counted with their multiplicity, that are neighbours to all vertices of $D$;  in other words, $S\bullet \Lambda_G^D$ is  the number of common neighbours of $D$ in $S$ ($.\bullet.$ is the scalar product: $A\bullet B = \sum_{u\in V}A(u).B(u)$, so $S\bullet \Lambda_G^D = \sum_{u \in \Lambda_G^D}S(u)$).

To any simple undirected graph $G=(V,E)$, is associated a quantum state $\ket G$, called graph state, defined as $$\ket G = \frac 1{\sqrt {2^n}}\sum_{x\in \{0,1\}^n}(-1)^{|G[x]|}\ket x$$ 
where $n$ is the order of $G$ and $|G[x]|$ denotes the number of edges in the subgraph of $G$ induced by $x$.\footnote{With a slight abuse of notation, $x\in \{0,1\}^n$ denotes the subset of vertices $\{u \in V | x_{\iota(u)} = 1\}$, where $\iota:V\to [0,n-1]$ s.t. $u\prec v\iff \iota(u)<\iota(v)$. $\iota$ is unique as $V$ is totally ordered by $\prec$.}

We are interested in the action of local unitaries on graph states. A local unitary is a tensor product of 1-qubit unitaries like Hadamard $H:\ket{a}\mapsto \frac{\ket 0+(-1)^a\ket1}{\sqrt 2}$, and Z- and X-rotations that are respectively defined as follows:
$$Z(\alpha)\!:=e^{i\frac \alpha2}\left(\cos\left(\frac \alpha2\right)I-i\sin\left(\frac \alpha2\right)Z\right); X(\alpha)\!:=HZ(\alpha)H=e^{i\frac \alpha2}\left(\cos\left(\frac \alpha2\right)I-i\sin\left(\frac \alpha2\right)X\right)$$
where  $X:\ket a\mapsto \ket {1-a}$ and $Z:\ket a \mapsto (-1)^a\ket a$.
Any 1-qubit unitary can be decomposed into $H$ and $Z(\alpha)$ rotations, whereas 1-qubit Clifford operators are those generated by $H$ and $Z(\frac \pi 2)$. Local complementation (denoted by the operator $\star$) can be implemented by local Clifford operators: 
$$\ket{G\star u} = X_u\left(\frac \pi {2}\right)\bigotimes_{v\in N_G(u)}Z_v\left(-\frac \pi {2}\right)\ket{G}$$

Conversely, if two graph states are related by local Clifford unitaries, the corresponding graphs are related by local complementations \cite{VandenNest04}. Thus, we use the term of LC-equivalence to describe both local Clifford equivalent graph states, and graphs related by local complementations (conveniently, LC stands for both \emph{local Clifford} and \emph{local complementation}). Similarly, we say that two graphs are LU-equivalent (resp. LC$_r$-equivalent) when there is a local unitary (resp. a local unitary generated by $H$ and $Z(\frac \pi {2^r})$) transforming the corresponding graph states into each other. 

Graph states form a subfamily of the well-known \emph{stabilizer states}, indeed $\ket G$ is the fix point of $X_u\bigotimes_{v\in N_G(u)}Z_v$ for any $u\in V$. When analysing the entanglement properties of stabilizer states, it is natural to focus on graph states as every stabilizer state is known to be local Clifford equivalent to a graph state \cite{VandenNest04}. Moreover, there are efficient procedures to associate with any stabilizer state an LC-equivalent graph state \cite{VdnEfficientLC}, thus the problem of deciding the LU-equivalence of stabilizer states naturally reduces to the LU-equivalence of graph states.

We describe in the next section a recent graphical characterisation of LU- and LC$_r$-equivalences of graph states based on the so-called generalised local complementation~\cite{claudet2024local}. 

\subsection{Generalised local complementation}

We review the definition of generalised local complementation and a few of its basic properties. The reader is referred to \cite{claudet2024local} for a more detailed introduction. A generalised local complementation is a graph transformation parametrised by an independent (multi)set of vertices $S$ and a positive number $r$ called \emph{level}. Like the usual local complementation, the transformation consists in toggling some of the edges of the graph depending on the number of neighbours the endpoint vertices have in common in $S$. Roughly speaking an $r$-local complementation toggles an edge if the number of common vertices in $S$ is an odd multiple of $2^{r-1}$ (an example of $2$-local complementation is given in \cref{fig:generalized_lc}). To be a valid $r$-local complementation, the (multi)set $S$ on which the transformation is applied should  be $r$-incident, i.e.~the number of common neighbours in $S$ of any set of at most $r$ vertices should be an appropriate power of two: 

\begin{definition}[$r$-Incidence]\label{def:r-inc}
    Given a graph $G$, a multiset $S$ of vertices is  $r$-incident, if for any $k\in [0,r)$, and any $K\subseteq V\setminus \supp(S)$ of size $k+2$, their number $S\bullet \Lambda_G^K$ of common neighbours in $S$ is a multiple of $2^{r-k-\delta(k)}$,
    where $\delta$ is the Kronecker delta\footnote{$\delta(x)\in \{0,1\}$ and $\delta(x)=1 \Leftrightarrow x=0$.
    }.
\end{definition}
 
\begin{definition}[$r$-Local Complementation]\label{def:r-LC}
    Given a graph $G$ and an $r$-incident independent multiset $S$, let $G\star^rS$ be the graph defined as
    $$u\sim_{G\star^r S} v ~\Leftrightarrow~\left(u\sim_{G} v ~~\oplus~~ S \bullet\Lambda_G^{u,v} = 2^{r-1}\bmod 2^{r}\right)$$
\end{definition}


\noindent
Below we recall some basic properties of generalised local complementation.

\begin{proposition}{\bf \cite{claudet2024local}}
    ~
    \begin{itemize}
        \item Generalised local complementations are self inverse: if $G\star^r S$ is valid, then $(G \star^r S ) \star^r S = G$.
        \item The multiplicity in $S$ can be upperbounded by $2^r$: if $G\star^r S$ is valid, then $G\star^r S= G\star^r S'$, where, for any vertex $u$, $S'(u)=S(u)\bmod 2^r$.
        \item If $G\star^r S_1$ and $G\star^r S_2$ are valid and  
        $S_1 + S_2$ is independent in $G$, then $G\star^r (S_1+ S_2) = (G\star^r S_1)\star^r S_2$. 
        \item If $G\star^r S$ is valid then $G\star^{r+1} (S + S)$ is valid and induces the same transformation: $G\star^{r+1} (S + S) = G\star^{r+1} S\star^{r+1} S = G\star^rS$.
        \item If $G\star^r S$ is valid then $G\star^{r-1} S$ is valid (when $r>1$) and  $G\star^{r-1} S = G$.
    \end{itemize}
\end{proposition}

\subsection{1- and 2-local complementation} 

To illustrate how $r$-local complementation behaves, we consider the simple cases $r=1$ and $r=2$. First, any multiset $S$ is $1$-incident, and a 1-local complementation is nothing but a sequence of usual local complementations. 2-local complementations cannot always be decomposed into usual local complementations, it is however sufficient to consider $2$-local complementations over sets, rather than multisets (see \cref{fig:generalized_lc}):

\begin{figure}[h]
    \centering
    \begin{tikzpicture}[xscale = 0.4,yscale=0.3]
    
    \begin{scope}[every node/.style={circle,minimum size=15pt,thick,draw,fill=lipicsYellow, inner sep = 0pt}]
        \node (U1) at (-5,5) {$a$};
        \node (U2) at (0,5) {$b$};
        \node (U3) at (5,5) {$c$};
        \node (U4) at (-5,0) {$d$};
        \node (U5) at (0,0) {$e$};
        \node (U6) at (5,0) {$f$};
    \end{scope}
    \begin{scope}[every node/.style={},
                    every edge/.style={draw=darkgray,very thick}]                   
        \path [-] (U1) edge node {} (U4);
        \path [-] (U1) edge node {} (U5);
        \path [-] (U1) edge node {} (U6);
        \path [-] (U2) edge node {} (U5);
        \path [-] (U2) edge node {} (U6);
        \path [-] (U3) edge node {} (U5);
        \path [-] (U3) edge node {} (U6);
    
        \path [-] (U4) edge node {} (U5);
        \path [-] (U5) edge node {} (U6);
    \end{scope}
    
    \begin{scope}[shift={(23,0)},every node/.style={circle,minimum size=15pt,thick,draw,fill=lipicsYellow, inner sep = 0pt}]
        \node (U1) at (-5,5) {$a$};
        \node (U2) at (0,5) {$b$};
        \node (U3) at (5,5) {$c$};
        \node (U4) at (-5,0) {$d$};
        \node (U5) at (0,0) {$e$};
        \node (U6) at (5,0) {$f$};
    \end{scope}
    \begin{scope}[every node/.style={},
                    every edge/.style={draw=darkgray,very thick}]                   
        \path [-] (U1) edge node {} (U4);
        \path [-] (U1) edge node {} (U5);
        \path [-] (U1) edge node {} (U6);
        \path [-] (U2) edge node {} (U5);
        \path [-] (U2) edge node {} (U6);
        \path [-] (U3) edge node {} (U5);
        \path [-] (U3) edge node {} (U6);
    
        \path [-] (U4) edge[bend right] node {} (U6);
        \path [-] (U5) edge node {} (U6);

    \end{scope}
    
    \draw[-Stealth, line width=2pt] (8.5,2.5) -- (14,2.5);
    \node (t) at (11.25,3.5) {\small 2-local complementation};
    \node (t2) at (11.25,1.5) {\small  over $\{a,a,b,c\}$};

    \end{tikzpicture}
    \caption{Illustration of a 2-local complementation over the multiset $S = \{a,a,b,c\}$. $S$ is 2-incident: indeed $S\bullet \Lambda_G^{\{d,e,f\}} = 2$, which is a multiple of $2^{2-1-0} = 2$. Similarly, $S\bullet \Lambda_G^{\{d,e\}} = S\bullet \Lambda_G^{\{d,f\}} = 2$ and $S\bullet \Lambda_G^{\{e,f\}} = 4$. Edges $de$ and $df$ are toggled as $S\bullet \Lambda_G^{\{d,e\}} = S\bullet \Lambda_G^{\{d,f\}} = 2\bmod 4$, but not edge $ef$ as $S\bullet \Lambda_G^{\{e,f\}} = 0\bmod 4$. Following \cref{prop:decomposition}, the 2-local complementation over $S$ can be decomposed into a 2-local complementation over the set $\{b,c\}$ and a 1-local complementation over the set $\{a\}$.}
    \label{fig:generalized_lc}    
    \end{figure}
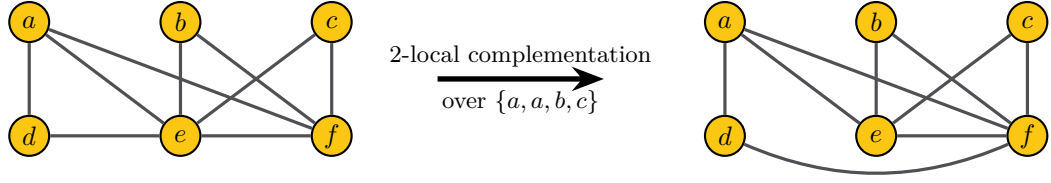

\begin{proposition} \label{prop:decomposition}
Any 2-local complementation can be decomposed into 1- and 2-local complementations over sets.
\end{proposition}
\begin{proof}
    Given a 2-local complementation over a multiset $S$, we assume without loss of generality that multiplicities are defined modulo $2^2 = 4$. Let  $S_1$ be the set of vertices that have multiplicity $2$ or $3$ in $S$. Notice that $G\star^2 S = G\star^2 S\star^1 S_1\star^1 S_1$ as $S_1$ is an independent set, $S_1$ is 1-incident and generalised local complementations are self inverse. So $G\star^2 S = G\star^2 S\star^2 (S_1 + S_1)\star^1 S_1 = G\star^2(S+ S_1 + S_1)\star^1 S_1$, where the multiplicity in 
    $S+ S_1 + S_1$ is either $0$ or $1$ modulo $4$. Thus the $2$-local complementation over the multiset $S$ can be decomposed into 2- and 1-local complementations over sets $\supp(S+ S_1+S_1)$ and $S_1$ respectively. 
\end{proof}

The 2-incidence condition can be rephrased as follows when $S$ is a set: for any subset $K$ of $V\setminus S$ of size $2$ or $3$, there is an even number of common neighbours in $S$: $|S\cap \Lambda_G^K|=0\mod 2$. In other words, the cut matrix describing the edges between $S$ and $V\setminus S$ is tri-orthogonal \cite{bravyi2012magic, shi2024, nezami2022}.  

\subsection{LU-equivalence and generalised local complementation}
 
While local complementation can be implemented on graph states by means of local Clifford unitaries, $r$-local complementations can be implemented on graph states with local unitaries generated by $H$ and $Z\left(\frac \pi {2^r}\right)$:

$$\ket{G\star^r S} = \bigotimes_{u\in V}X\left(\frac {S(u)\pi}{2^r}\right)\bigotimes_{v\in V}Z\left(-\frac {\pi}{2^r}\sum_{u \in N_G(v)}S(u)\right)\ket{G}$$
Conversely, if two graph states are related by local unitaries generated by $H$ and $Z\left(\frac \pi {2^r}\right)$, the corresponding graphs are related by $r$-local complementations \cite{claudet2024local}. In other words two graphs are LC$_r$-equivalent if and only if there is a series of $r$-local complementations transforming one into another. Two LC$_r$-equivalent graphs are also LC$_{r+1}$-equivalent, however the converse does not hold, resulting in an infinite strict hierarchy of local equivalences \cite{claudet2024local}. 
Most importantly, generalised local complementations capture the LU-equivalence of graphs:

\begin{theorem}{\bf \cite{claudet2024local}}\label{thm:LU_imply_LCr}
    If $G_1$ and $G_2$ are LU-equivalent, then $G_1$ and $G_2$ are LC$_{\lfloor n/2 \rfloor-1}$ equivalent, where $n$ is the order of the graphs, i.e.~there exists a sequence of $({\lfloor n/2 \rfloor-1})$-local complementations transforming $G_1$ into $G_2$.
\end{theorem}

Finally, a peculiar property is that for any pair of LU-equivalent graphs, a single generalised complementation, together with usual local complementations, is sufficient to transform one graph into the other:

\begin{proposition}{\bf \cite{claudet2024local}} \label{prop:LCr_lc}
 If $G_1$ and $G_2$ are LC$_r$-equivalent, then $G_1$ and $G_2$ are related by a sequence of generalised local complementations, such that a single one is of level $r$, all the others are usual local complementations (i.e.~level $1$).

\end{proposition}

\section{Algorithms for LC$_r$- and LU-equivalences}

In this section, we address the problem of deciding whether two given graphs are LU-equivalent. Additionally, we  consider a variant of this problem consisting in deciding whether two graphs are LC$_r$-equivalent for a fixed $r$.  Since LU-equivalent graphs are necessarily LC$_r$-equivalent for some $r$, the difference lies in whether the level $r$ is fixed or not. 

Thanks to \cref{prop:LCr_lc}, if $G_1$ is  LC$_r$-equivalent to $G_2$, there exists a single $r$-local complementation, together with usual local complementations, that transforms $G_1$ into $G_2$. We introduce an algorithm that builds such a sequence of generalised local complementations, in essentially four stages: 
\begin{itemize}
\item[(i)] Both $G_1$ and $G_2$ are turned in \emph{standard forms} $G_1'$ and $G'_2$ by means of (usual) local complementations. These transformations are driven by a so-called  minimal local set cover which can be efficiently computed.  
\item[(ii)] We then focus on the single $r$-local complementation: all the possible actions of a single $r$-local complementation on $G'_1$ are described as a vector space, for which we compute a basis $\mathcal B$.  
\item[(iii)] It remains to find, if it exists, the $r$-local complementation to apply on $G'_1$ that leads to $G'_2$ up to some additional usual local complementations. With an appropriate construction depending on $G_1'$, $G'_2$ and $\mathcal B$, we reduce this problem to deciding whether two graphs are LC-equivalent under some additional requirements  on the sequence of local complementations to apply. These requirements can be expressed as linear constraints.
\item [(iv)] Finally, to find such a sequence of local complementations, we apply a variant of Bouchet's algorithm, generalised to accommodate the additional linear constraints. 
\end{itemize}

Stages $(i)$, $(iii)$ and $(iv)$ can be performed in polynomial time in the order $n$ of the graphs. Stage $(ii)$ has essentially a $O(n^r)$ time complexity, thus deciding LC$_r$-equivalence for a fixed $r$ can be done in polynomial time. Regarding LU-equivalence, \cref{thm:LU_imply_LCr} implies $r\le \frac n2$. We improve this upperbound and show that $r$ is at most logarithmic in $n$, leading to a quasi-polynomial time algorithm for LU-equivalence.

The rest of this section is dedicated to the description of the algorithm, its correctness and complexity, beginning with the generalisation of Bouchet's algorithm to decide, in polynomial time, LC-equivalence with additional constraints.

\subsection{Bouchet algorithm, revisited} \label{subsec:bouchet}

LC-equivalence can be efficiently decided thanks to the famous Bouchet's algorithm \cite{Bouchet1991}. Bouchet proved that LC-equivalence of two given graphs, defined on the same vertex set $V$, reduces to the existence of subsets of vertices satisfying the following two equations:

\begin{proposition}{\bf \cite{Bouchet1991}} \label{prop:Bouchet}
Two graphs $G$, $G'$ are LC-equivalent if and only if there exist $A,B,C,D\subseteq V$ such that
\begin{itemize}
\item[(i)] $\forall u,v \in V$,\\
$|B\cap N_G(u)\cap N_{G'}(v)| + |A\cap N_G(u)\cap \{v\}| +  |D\cap \{u\}\cap N_{G'}(v)| + |C\cap \{u\}\cap \{v\}| = 0\bmod 2$
\item[(ii)]
$(A\cap D)~\Delta~ (B\cap C) = V$ 
\end{itemize}

\end{proposition}

While the original proof involves isotropic systems~\cite{bouchet1987}, we provide an alternative, self-contained proof in \cref{app:Bouchet}, that we believe to be more accessible than the original one.  

Notice that Equation $(i)$ is  actually a linear equation: the set $\mathcal S\subseteq V^4$ of solutions to  $(i)$ is a vector space, indeed given two solutions $S=(A,B,C,D)$ and $S'=(A',B',C',D')$ of $(i)$, so is $S+S'=(A\Delta A', B\Delta B', C\Delta C', D\Delta D')$. The linearity of equation $(i)$ can be emphasised using the following encoding: 

A set $A\in V$ can be represented by $n$ binary variables $a_1, \ldots, a_n\in \mathbb F_2$ s.t. $a_v=1 \Leftrightarrow v\in A$, moreover, with a slight abuse of notations, we identify any set $A\subseteq V$ with the corresponding diagonal $\mathbb F_2$ matrix of dimension $n\times n$ in which diagonal elements are the $(a_v)_{v\in V}$. Following \cite{VdnEfficientLC,Hein06}, equation $(i)$ is equivalent to 
\begin{equation}
\Gamma  B \Gamma' + \Gamma  A+  D\Gamma' +  C = 0
\end{equation}
and equation $(ii)$ to 
\begin{equation}
AD+BC=I
\end{equation}
where $\Gamma$ and $\Gamma'$ are the adjacency matrices of $G$ and $G'$ respectively.

In this section, we consider an extension of Bouchet's algorithm where an additional set of linear constraints on $A,B,C$ and $D$ is added as input of the problem. Such additional linear equations can reflect constraints on the applied  local complementations, e.g.~deciding whether two graphs are LC-equivalent under the additional constraint that all local complementations are applied on a fixed set $V_0$ of vertices (see \cref{ex:V0}).

While solving linear equations is easy, equation $(ii)$ is not linear, and the tour de force of Bouchet's algorithm is to point out the fundamental properties of the solutions to both equations $(i)$ and $(ii)$ that allow to decide efficiently the LC-equivalence of graphs.  In particular, Bouchet showed that a set of solutions $\mathcal C\subseteq \mathcal S$ that satisfies both $(i)$ and $(ii)$,  is either small or it contains an affine subspace of $\mathcal S$ of small co-dimension. In  the latter case, the entire set $\mathcal S$ is  actually an affine space except for  some particular cases that can be avoided by assuming that the graphs contain vertices of even degree. We extend this result as follows:

\begin{lemma}\label{lemma:codim2}
Given $G$, $G'$ two connected graphs with at least one vertex of even degree, and a set $L$  of linear constraints on $V^4$,
 then either the set $\mathcal S_L$ of solutions to both $ L$ and $(i)$ is of dimension at most $4$, 
 or the set  $\mathcal C_L\subseteq \mathcal S_L$ that additionally satisfies  $(ii)$ is either empty or an affine subspace of $\mathcal S_L$ of codimension at most $2$.   
\end{lemma}

\begin{proof} It is enough to consider the case $\dim(\mathcal S_L)>4$ and $\mathcal C_L\neq \emptyset$. The set  $\mathcal S$ of solutions to $(i)$ contains $\mathcal S_L$, so $\dim(\mathcal S)>4$. According to \cite{Bouchet1991}, $\mathcal C$, the solutions to $(i)$ and $(ii)$, is an affine subspace of $\mathcal S$ of codimension at most $2$. For any  $a\in \mathcal C_L$, $\mathcal C = a+ \mathcal N$ where $ \mathcal N$ is a subvector space of $\mathcal S$. Notice that $\mathcal C_L = \mathcal C\cap \mathcal S_L = a+ \mathcal N\cap \mathcal S_L $. We have $\dim(\mathcal C_L)=\dim(\mathcal N\cap \mathcal S_L ) = \dim(\mathcal N)+\dim(\mathcal S_L )-\dim(\mathcal N + \mathcal S_L )\ge \dim(\mathcal N)+\dim(\mathcal S_L )-\dim(\mathcal S)\ge \dim(\mathcal S_L)-2$ as $\mathcal N$ is of codimension at most $2$ in $\mathcal S$. 
\end{proof}

\begin{remark} \cref{lemma:codim2} holds actually for any graph that is not in the so-called `Class $\alpha$' of graphs with only odd-degree vertices together with a few additional properties\footnote{(a) any pair of non adjacent vertices should have an even number of common neighbours ; (b) for any cycle $C$, the number of triangles having a edge in $C$ is equal to the size of $C$ modulo $2$.}.  When the graphs are in  `Class $\alpha$', and in the absence of additional constraints, Bouchet proved that there is at most $2$ solutions in $\mathcal C$ that do not belong to the affine subspace of small codimension, and these two solutions can be easily computed (see \cite{Bouchet1991}, section 7).  We leave as an open question the description of the set of solutions for graphs in `Class $\alpha$', in particular when the two particular solutions pointed out by Bouchet do not satisfy $L$, the set of additional constraints. 
\end{remark}
From an algorithmic point of view, \cref{lemma:codim2} leads to a straightforward generalisation of Bouchet's algorithm to  efficiently decide LC-equivalence of graphs, under a set of additional linear constraints: 

\begin{proposition}\label{prop:extendedBouchet}
Given $G$, $G'$ two connected graphs of order $n$ with an even-degree vertex, and a set $L$ of $\ell$ linear constraints on $V^4$, one can compute a solution to both $(i)$, $(ii)$ and $L$ when it exists, or claim there is no solution, in runtime $O((n^2+\ell)n^2)$.
\end{proposition}

\begin{proof} 
A Gaussian elimination can be used to compute a basis $\mathcal B = \{S_0, \ldots , S_k\}$ of $\mathcal S_L$, with $k<n$, in $O((n^2+\ell)n^2)$ operations as there are $n^2$ equations in $(i)$. If the dimension of $\mathcal S_L$ is at most $4$ (so $|\mathcal S_L|\le 16$), we check in $O(n)$ operations, for each element of $\mathcal S_L$ whether equation $(ii)$ is satisfied. Otherwise, when $\dim(\mathcal S_L)>4$, if $\mathcal C_L$ is non empty, at least one element of $\mathcal C_L$ is the sum $S_i +  S_j$ of two  elements of $\mathcal B$ (see Lemma 4.4 in \cite{Bouchet1991}). For each of the $O(n^2)$ candidates we check  whether condition $(ii)$ is satisfied.  If no solution is found, it implies that $\mathcal C_L=\emptyset$. 
\end{proof}

As the algorithm described in \cref{app:Bouchet} translates a solution to (i) and (ii), into a sequence of local complementations relating two graphs, some constraints on the sequence of local complementations may be encoded as additional linear constraints. We give a fairly simple example below. A more intricate example is presented in \cref{lemma:LC_new_graphs} (in \cref{subsec:algorithm_lcr}).

\begin{example} \label{ex:V0}
    Let $G$, $G'$ be two connected graphs of order $n$ with an even-degree vertex, and $V_0$ a set of vertices. One can decide in runtime $O(n^4)$ whether there exists a sequence of (possibly repeating) vertices $a_1, \cdots, a_m\in V_0 $ such that $G' = G \star a_1 \star \cdots \star a_m$. Roughly speaking, the idea is to consider the linear constraint $b_u=0$ (i.e.~$u\in   \overline B$) for any $u\notin V_0$, to reflect the constraints that local complementations should not be applied outside of $V_0$.
\end{example}

An interpretation of the possible additional constraints of the extended Bouchet algorithm in terms of local Clifford operators over graph states is given in \cref{app:interpretationClifford}.

\subsection{Minimal local sets and standard form} \label{subsec:mls}

We consider in this section the first stage of the LU-equivalence algorithm, consisting in putting the two input graphs into a particular shape called standard form by means of local complementations. We adapt a transformation introduced in \cite{claudet2024local}, which is based on the so-called minimal local sets, 
and turn it into an efficient procedure.

\begin{definition} A \emph{local set} $L$ is a non-empty subset of $V$ of the form $L = D \cup Odd_G(D)$ for some $D \se V$ called a \emph{generator}. A minimal local set is a local set that is minimal by inclusion.
\end{definition}

A key property of local sets is that they are invariant under LU-equivalence: Two LU-equivalent graphs $G_1$, $G_2$ have the same local sets, but not necessarily with the same generators. Moreover, the way the generators of a minimal local set differ in $G_1$ and $G_2$, 
provides some information on the sequence of generalised local complementations that transforms $G_1$ into $G_2$. 
It is thus important to cover all vertices of a graph with at least one minimal local set. Fortunately, any graph admits a minimal local set cover (MLS cover for short), and an MLS cover can be computed efficiently, within   $O(n^{6.38})$ operations\footnote{The algorithm presented in \cite{claudet2024covering} computes a MLS cover in $O(n^4)$ evaluations of the so-called cut-rank function, which itself can be computed in $O(n^\omega)$ field operations where $\omega < 2.38$.} where $n$ is the order of the graph \cite{claudet2024covering}. The information that an MLS cover provides on each vertex, is reflected by a type X, Y, Z or $\bot$, defined as follows:

\begin{definition} \label{def:type}Given a graph $G$, a vertex $u$ is of type P $\in$ \{X, Y, Z, $\bot$\} with respect to a MLS cover $\mathcal M$, where P is
    \begin{itemize}
        \item  X if for any generator $D$ of a minimal local set of $\mathcal M$ containing $u$, $u\in D \sm Odd(D)$,
        \item  Y if for any generator $D$ of a minimal local set of $\mathcal M$ containing $u$, $u\in D \cap Odd(D)$, 
        \item Z if for any generator $D$ of a minimal local of $\mathcal M$ set containing $u$, $u\in Odd(D) \sm D$, 
        \item $\bot$ otherwise. 
    \end{itemize}
    \end{definition}
        
When a local complementation is applied on a vertex $u$, its type remains unchanged if it is X or $\bot$, while types Y and Z are swapped. For the neighbours of $u$, types Z and  $\bot$ remain unchanged, whereas  X and  Y are exchanged. This leads to a notion of standard form:

\begin{definition}
    A graph $G$ is in standard form with respect to a MLS cover $\mathcal M$ if
    \begin{itemize}
        \item There are no vertices of type Y with respect to $\mathcal M$,
        \item For every vertex $u$ of type X with respect to $\mathcal M$, any neighbour $v$ of $u$ is of type Z with respect to $\mathcal M$ and satisfies $u\prec v$, in particular the vertices of type X with respect to $\mathcal M$ form an independent set,
        \item For every vertex $u$ of type X with respect to $\mathcal M$, $\{u\} \cup N_G(u) \in \mathcal M$.
    \end{itemize}
\end{definition}

\begin{remark}This notion of standard form is a generalisation of the one introduced in \cite{claudet2024local}, where the MLS cover considered consists of every minimal local set of the graph: $\mathcal M_\text{max} \defeq \{L \se V ~|~ L \text{~is a minimal local set}\}$. Since there can be exponentially many minimal local sets, using $\mathcal M_\text{max}$ does not lead to an efficient procedure, for instance when computing the type of each vertex.
\end{remark}

Given a pair of LU-equivalent graphs, one can efficiently compute a (common) MLS cover and put both graphs in standard form by means of local complementations:

\begin{restatable}{lemma}{standardform} \label{lemma:standardform}
    There exists an efficient algorithm that takes as inputs two graphs $G_1$ and $G_2$ of order $n$, and either claim that they are not LU-equivalent, or compute an MLS cover $\mathcal M$ and two graphs $G'_1$ and $G'_2$ LC-equivalent to $G_1$ and $G_2$ respectively, such that $G'_1$ and $G'_2$ are both in standard form with respect to $\mathcal M$, in runtime $O(n^{6.38})$.
\end{restatable}

The algorithm is fairly similar to the one presented in the proof of Proposition 24 in \cite{claudet2024local} (the main difference being that we may now add minimal local sets to the MLS cover) and can be found in \cref{app:standardform}. Notice the most computationally expensive step of the algorithm is the computation of the MLS cover, hence the runtime $O(n^{6.38})$. Standard forms with respect to a common MLS cover implies some strong similarities in the structure of graphs:

\begin{lemma} \label{lemma:same_types}
    If two graphs $G_1$ and $G_2$ are LU-equivalent and in standard form with respect to an MLS cover $\mathcal M$, then every vertex has the same type in $G_1$ and $G_2$, and  every vertex $u$ of type $X$ satisfies $N_{G_1}(u) = N_{G_2}(u)$. 
\end{lemma}

\cref{lemma:same_types} was proved in \cite{claudet2024local} for the maximal MLS cover $\mathcal M_\text{max}$, but the mathematical arguments hold for any arbitrary MLS cover. A key argument is that two LU-equivalent graphs have the same vertices of type $\bot$ with respect to any arbitrary MLS cover.

After performing the algorithm described in \cref{lemma:standardform}, one can check in quadratic time\footnote{In the order of the graphs, assuming the information of the types of the vertices with respect to the MLS cover is conserved.} whether each vertex has the same type  in $G_1$ and $G_2$, and whether every vertex of type X has the same neighbourhood in both graphs. If either condition is not met, the graphs are not LU-equivalent.

Finally, thanks to standard form, deciding LC$_r$-equivalence of graphs reduces to determining whether they are related by a single $r$-local complementation along with some usual local complementations:
\begin{lemma} \label{lemma:standardform_LCr_lc}
    If $G_1$ and $G_2$ are LC$_r$-equivalent and are both in standard form with respect to an MLS cover $\mathcal M$, then $G_1$ and $G_2$ are related by a sequence of local complementations on the vertices of type $\bot$ along with a single $r$-local complementation over the vertices of type X. 
\end{lemma}

\cref{lemma:standardform_LCr_lc} was proved in \cite{claudet2024local} for the maximal MLS cover $\mathcal M_\text{max}$, but the mathematical arguments hold for any arbitrary MLS cover.

\subsection{An algorithm to recognise \texorpdfstring{LC$_r$}{LCr}-equivalent graph states} \label{subsec:algorithm_lcr}

We are now ready to describe the algorithm that recognises two LC$_r$-equivalent graphs. We consider a level $r \gs 1$, and two graphs $G_1$ and $G_2$ of order $n$, defined on the same vertex set $V$. Following \cref{lemma:standardform}, assume, without loss of generality, that $G_1$ and $G_2$ are both in standard form with respect to the same MLS cover. Then, it is valid (see \cref{lemma:same_types}) to define $V_X, V_Z \se V$, the sets of vertices respectively of type X and Z with respect to the MLS cover. Also, each vertex in $V_X$ has the same neighbourhood in both $G_1$ and $G_2$. According to \cref{lemma:standardform_LCr_lc}, if $G_1$ and $G_2$ are LC$_r$-equivalent then there is a single $r$-local complementation over vertices of $V_X$ together with a series of local complementations on vertices of type $\bot$ that transform $G_1$ into $G_2$. We first focus on the single $r$-local complementation (that commutes with the local complementations on vertices of type $\bot$, as there is no edge between a vertex of type X and a vertex of type $\bot$) and thus consider all the possible graphs that can be reached from $G_1$ by mean of a single $r$-local complementation over vertices of $V_X$. 
Notice that such a $r$-local complementation only toggles edges which both endpoints are in $V_Z$. Given a multiset $S$, the edges toggled in $G_1\star^r S$ can be represented by a vector $\omega^{(S)}\in {{\mathbb F}_2}^{\{u,v \in V_Z ~|~ u \neq v\}}$ such that for any $u,v \in V_Z$, $\omega^{(S)}_{u,v} = u\sim_{G_1} v ~\oplus~ u\sim_{G_1\star^r S} v$. The actions of all the possible $r$-local complementations on $G_1$ can thus be described as the set $\Omega =\{\omega^{(S)}  |  \text{$S$ is an $r$-incident multiset of vertices of type X}\}$.

\begin{lemma} \label{lemma:omega}
    $\Omega$ is a vector space and a basis $\mathcal B$ of $\Omega$ can be computed in runtime $O(r n^{r+2.38})$.
\end{lemma}

\begin{proof}
    A multiset $S$ of vertices of type X can be represented as a vector in $({\Zp{2^r}})^{V_X}$ which entries are $S(u) \bmod 2^{r}$, as the multiplicity can be considered modulo $2^r$ in the context of a $r$-local complementation. Let $\Sigma$ be the subset of $({\Zp{2^r}})^{V_X}$ that corresponds to all $r$-incident independent multisets of vertices in $V_X$. $\Sigma$ is a vector space since the property of $r$-incidence is preserved under the addition of two vectors. $\Sigma$ is actually the space of solutions to a set of $O(n^{r+1})$ equations\footnote{There are precisely $\binom{|V_Z|}{r+1}+\binom{|V_Z|}{r}+\cdots+\binom{|V_Z|}{3}+\binom{|V_Z|}{2}$ equations.} given by the conditions of $r$-incidence. With some multiplications by powers of 2, all these equations are expressible as equations modulo $2^{r}$. For every set  $K\subseteq V_Z$ of size between $2$ and $r+1$, the corresponding equation is $$\sum_{u \in \Lambda_{G_1}^K}2^{|K|-2+\delta(|K|-2)}S(u) = 0\bmod 2^{r}$$
    Notice that if $r=1$, the space of solutions is the entire space $({\Zp{2^r}})^{V_X}$, as there are no incidence constraints on usual local complementations. Summing up, to compute a basis of $\Sigma$, we solve a system of $O(n^{r+1})$ equations modulo $2^{r}$ with $O(n)$ variables. One can obtain a generating set $\{S_1,S_2,...,S_t\}$ of $\Sigma$ of size $t \ls n$ in $O(r n^{r+2.38})$ basic operations using an algorithm based on the Howell transform \cite{storjohann2000algorithms}. 

    Let $f$ be the function that associates with each element $S \in \Sigma$, its action $f(S) \in \Omega$ on the edges: for any $u,v \in V_Z$, $f(S)_{u,v} = u\sim_{G_1} v ~\oplus~ u\sim_{G_1\star^r S} v$. Notice that $f$ is linear, as for any $S, S' \in \Sigma$:
    \begin{align*}
        f(S+S') &= u\sim_{G_1} v ~\oplus~ u\sim_{G_1\star^r (S+S')} v
        = u\sim_{G_1} v ~\oplus~ u\sim_{(G_1\star^r S) \star^r S'} v\\
        &= u\sim_{G_1} v ~~\oplus~~ u\sim_{G_1\star^r S} v ~~\oplus~~ \left(S' \bullet\Lambda_{G_1\star^r S}^{u,v} = 2^{r-1}\bmod 2^{r}\right)\\
        &= u\sim_{G_1} v ~~\oplus~~ u\sim_{G_1\star^r S} v ~~\oplus~~ \left(S' \bullet\Lambda_{G_1}^{u,v} = 2^{r-1}\bmod 2^{r}\right)\\
        &= u\sim_{G_1 \star^r S} v ~~\oplus~~ u\sim_{G_1\star^r S'} v\\
        &= u\sim_{G_1} v ~\oplus~ u\sim_{G_1 \star^r S} v ~~\oplus~~ u\sim_{G_1} v ~\oplus~ u\sim_{G_1\star^r S'} v = f(S)+f(S')
    \end{align*} 

    This directly implies that $\Omega$ is a vector space: if $\omega, \omega' \in \Omega$, by definition there exists $S, S' \in \Sigma$ such that $f(S) = \omega$ and $f(S') = \omega$, moreover $\omega + \omega' = f(S) + f(S') = f(S+S') \in \Omega$. Also, let us prove that $\Omega$ is generated by $\{f(S_1),f(S_2),...,f(S_t)\}$. Take a vector $\tilde \omega \in \Omega$, by definition there exists $\tilde S \in \Sigma$ such that $\tilde \omega = f(\tilde S)$. $\tilde S$ can be expressed as a linear combination of vectors from the generating set, i.e.~$\tilde S = \sum_{i \in [1,t]}a_i S_i$ where $a_i \in \Zp{2^r}$. Then, for any $u,v \in V_Z$, $ (\tilde \omega)_{u,v} = f\left(\sum_{i \in [1,t]}a_i S_i\right)_{u,v} = \sum_{i \in [1,t]}a'_i f(S_i)_{u,v}$ by linearity of $f$, where $a'_i \in \mathbb F_2$ such that $a'_i = a_i \bmod 2$ when $a_i$ and $a'_i$ are viewed as integers. In other words, $\tilde w = \sum_{i \in [1,t]}a'_i f(S_i)$, implying that $\{f(S_1),f(S_2),...,f(S_t)\}$ is a generating set of $\Omega$. Using Gaussian elimination, one can easily obtain a basis $\mathcal B$ of $\Omega$ from $\{f(S_1),f(S_2),...,f(S_{t})\}$.

\end{proof}

Thanks to the exhaustive description of all possible $r$-local complementations on $G_1$, we are now ready to reduce LC$_r$-equivalence to LC-equivalence with some additional constraint. We denote by $G_1^{\#}$ (resp. $G_2^{\#}$) the graph obtained from $G_1$ (resp. $G_2$) by the following procedure. First, remove the vertices of $V_X$. Then, for each vector $\omega \in \mathcal B$, for each $u,v \in V_Z$ such that $\omega_{u,v} = 1$, add a vertex connected only to $u$ and $v$ and call it $p_{u,v}^{\omega}$, and let $\mathcal P^\omega= \{p_{u,v}^{\omega} ~|~ \omega_{u,v}=1\}$. In the following, we refer to the vertices added by this procedure as "new vertices".

\begin{restatable}{lemma}{newgraphs} \label{lemma:new_graphs}
    $G_1$ and $G_2$ are LC$_r$-equivalent if and only if there exists a sequence of (possibly repeating) vertices $a_1, \cdots, a_m$ such that $G^{\#}_2 = G^{\#}_1 \star a_1 \star \cdots \star a_m$ satisfying the following additional constraints:
    \begin{itemize}
        \item the sequence contains no vertex of $V_Z$;
        \item for each $\omega \in \mathcal B$, either the sequence contains every vertex of $\mathcal P^{\omega}$ exactly once, or it contains none. 
    \end{itemize}
\end{restatable}

The proof of \cref{lemma:new_graphs} makes use of \cref{lemma:standardform_LCr_lc} and is given in \cref{app:new_graphs}. 

There exists an efficient algorithm that decides whether two graphs are LC-equivalent with such additional constraints using our generalisation of Bouchet's algorithm.

\begin{lemma} \label{lemma:LC_new_graphs}
    Deciding whether there exists a sequence of (possibly repeating) vertices $a_1, \cdots, a_m$ such that $G^{\#}_2 = G^{\#}_1 \star a_1 \star \cdots \star a_m$, satisfying the additional constraints described in \cref{lemma:new_graphs}, can be done in runtime $O(n^4)$.
\end{lemma}

\begin{proof}
    If the vector space $\Omega$ is of dimension zero (i.e.~$\Omega$ only contains the null vector), then there is no additional constraint, thus one can apply the usual Bouchet algorithm that decides LC-equivalence of graphs.

    If $\Omega$ is not of dimension zero, then $G^{\#}_1$ and $G^{\#}_2$ both have at least one even-degree vertex (since every "new vertex" is of degree 2). Using the notations of \cref{subsec:bouchet}, let us define the following linear constraints on $V^4$:  $\forall u \in V_Z$, $\forall \omega \in \mathcal B$, $\forall v,v'\in \mathcal P^{\omega}$, 
    \begin{itemize}
        \item $ u \in \overline B$;
        \item $v \in \overline C$;
        \item $v \in B$ if and only if $v' \in B$.
    \end{itemize} 
    According to \cref{prop:extendedBouchet}, a solution to the system of equations composed of (i), (ii) (see \cref{prop:Bouchet}) and the additional linear constraints can be computed in runtime $O(n^4)$ when it exists. Following the algorithm in \cref{app:Bouchet}, such a solution yields a sequence of local complementations satisfying the additional constraints described in \cref{lemma:new_graphs}. Conversely, such a sequence of local complementations can be converted into a valid solution to the system of equations.
\end{proof}

Summing up, we have an algorithm that decides, for a fixed level $r$,  the LC$_r$-equivalence of graphs in polynomial runtime.

\begin{theorem} \label{thm:algolcr}
    There exists an algorithm that decides if two graphs are LC$_r$-equivalent with runtime $O(r n^{r+2.38} + n^{6,38})$, where $n$ is the order of the graphs.
\end{theorem}

The algorithm reads as follows:
\begin{enumerate}
    \item Put $G_1$ and $G_2$ in standard form with respect to the same MLS cover if possible, otherwise output NO.
    \item Check whether each vertex has the same type in $G_1$ and $G_2$, and whether every vertex of type X has the same neighbourhood in both graphs, otherwise output NO.
    \item Compute a basis of the vector space $\Omega$.
    \item Compute the graphs ${G}^{\#}_1$ and ${G}^{\#}_2$.
    \item Decide whether ${G}^{\#}_1$ and ${G}^{\#}_2$ are LC-equivalent with the additional constraints described in \cref{lemma:new_graphs}. Output YES if this is the case,  NO otherwise.
\end{enumerate}

Notice that the algorithm is exponential in $r$, in particular it does not provide an efficient algorithm to decide LU-equivalence of graph states. To address this issue, we provide in the next subsection some upper bounds on the level of a generalised local complementation.

\subsection{Bounds for generalised local complementation} \label{subsec:bounds}

In this section, we prove an upper bound on the level of a valid generalised local complementation: roughly speaking we show that if $G\star^r S$ is valid then $r$ is at most logarithmic in the order of the graph $G$. This bound is however not true in general as it has been shown in \cite{claudet2024local} that whenever $G\star^r S$ is valid, we have $G\star^r S = G\star^{r+1} (S+S)$. To avoid these pathological cases, we thus focus on genuine $r$-incident independent multisets:

\begin{definition}
Given a graph $G$, a $r$-incident independent multiset $S$ is \emph{genuine} if there exists a set $K \se V \sm \supp(S)$ such that $|K|>1$ and $\sum_{N_{G}(u)=K}S(u)$ is odd\footnote{With a slight abuse of notation, $\sum_{N_{G}(u)=K}S(u)$ is the sum over all $u\in V$ s.t. $N_{G}(u)=K$.}. 
\end{definition}

\begin{proposition} \label{prop:nontrivial}
If $G\star^r S$ is valid and there is no $S'$ such that $G\star^{r-1} S' = G \star^r S$ then $S$ is a genuine $r$-incident independent multiset.    
\end{proposition}

\begin{proof} By contradiction assume $S$ is an $r$-incident independent multiset that is not genuine. Let $S'$ be the multiset obtained from $S$ by choosing, for every set $K \se V \sm \supp(S)$ s.t. $\{u \in \supp(S)~|~N_{G}(u)=K\}$ is not empty, a single vertex $u \in \supp(S)$ s.t. $N_{G}(u)=K$, and setting $S'(u)=\sum_{N_{G}(u)=K}S(u)$ and for any other vertex $v \in \supp(S)$ s.t. $N_{G}(v)=K$, $S'(v)=0$. It is direct to show that $S'$ is $r$-incident and that $G\star^r S = G\star^r S'$. Then, let $S'/2$ be the multiset obtained from $S'$ by dividing by 2 the multiplicity of each vertex in $\supp(S')$. It is direct to show that $S'$ is $(r-1)$-incident and that $G\star^r S= G\star^{r-1} S'/2$.   
\end{proof}

Genuine $r$-incidence can only occur for multisets whose support is of size at least exponential in $r$.

\begin{lemma} \label{lemma:exp_support}
    If $r>1$ and $S$ is a genuine $r$-incident independent multiset of a graph $G$, then $|\supp(S)| \gs 2^{r+2}-r-3$.
\end{lemma}

\begin{proof}
    Let $m>1$ be the smallest integer such that there exists a set $K_0 \se V \sm \supp(S)$ of size $m$ such that $S \bullet\Lambda_G^{K_0}$ is odd.   
    Note that by hypothesis there exists such an integer. Indeed, let $K_\text{max}$ be the biggest subset (by inclusion) of $V \sm \supp(S)$ such that $\sum_{N_{G}(u)=K_\text{max}}S(u)$ is odd: then $S \bullet\Lambda_G^{K_\text{max}}$ is odd. Thus, by definition of the $r$-incidence, $m\gs r+2$.

    Let $G' = G[\supp(S) \cup K_0]$ the graph obtained from $G$ by removing the vertices that are neither in the support of $S$, nor in $K_0$. By definition, $S$ is also $r$-incident in $G'$. Also, $S \bullet\Lambda_{G'}^{K_0}$ is odd, and for every set $K \subsetneq K_0$ s.t. $|K|>1$, $S \bullet\Lambda_{G'}^{K}$ is even.

    Let us prove that for any $K \se K_0$ s.t. $|K|>1$, $\sum_{N_{G'}(u)=K}S(u)$ is odd, by induction over the size of $K$. First notice that $\sum_{N_{G'}(u)=K_0}S(u) = S \bullet\Lambda_{G'}^{K_0}$ is odd. Then, let $K_1 \subsetneq K_0$ s.t. $|K_1|>1$.
    \begin{align*}
        &S \bullet\Lambda_{G'}^{K_1} = \sum_{K_1\se K\se K_0}\sum_{N_{G'}(u)=K}S(u)
        = \sum_{N_{G'}(u)=K_1}S(u) + \sum_{K_1\subsetneq K\se K_0}\sum_{N_{G'}(u)=K}S(u)\\
        &= \sum_{N_{G'}(u)=K_1}S(u) + |\{K\se K_0~|~K_1\subsetneq K\}| \bmod 2 \text{~~by hypothesis of induction}\\
        &= \sum_{N_{G'}(u)=K_1}S(u) + 1 \bmod 2 
    \end{align*}
Thus, $\sum_{N_{G'}(u)=K_1}S(u)$ is odd. As a consequence, for any $K \se K_0$ s.t. $|K|>1$, there exists at least one vertex $u\in \supp(S)$ s.t. $N_{G'}(u)=K$. Then, $|\supp(S)| \gs |\{K\se K_0~|~|K|>1\}| = 2^m - m - 1 \gs 2^{r+2} - (r+2) -1 = 2^{r+2}-r-3$.
\end{proof}

Likewise, $r$-local complementations that cannot be implemented by $(r-1)$-local complementations can only occur or multisets with sufficiently many vertices outside their support.

\begin{restatable}{lemma}{sizeZ}\label{lemma:sizeZ}
    If $G\star^r S$ is valid and there is no $S'$ such that $G\star^{r-1} S' = G \star^r S$, then $|V\sm\supp(S)| \gs r+3$.
\end{restatable}

The proof of \cref{lemma:sizeZ} involves similar techniques as the proof of \cref{lemma:exp_support} and is given in \cref{app:sizeZ}. \cref{lemma:exp_support,lemma:sizeZ} together give a simple bound involving only the order of the graph.

\begin{proposition} \label{prop:boundlevel}
    If $G\star^r S$ is valid and there is no $S'$ such that $G\star^{r-1} S' = G \star^r S$, then $n \gs 2^{r+2}$, where $n$ is the order of $G$.
\end{proposition}

Put differently, any $r$-local complementation on a graph of order at most $2^{r+2}-1$ can be implemented by $(r-1)$-local complementations:

\begin{corollary} \label{cor:LCr_LCr-1}
    If two graphs of order at most $2^{r+2}-1$ are LC$_r$-equivalent, then they are LC$_{r-1}$-equivalent.
\end{corollary}

In other words, two LC$_r$-equivalent but not LC$_{r-1}$-equivalent graphs are of order at least $2^{r+2}$. This implies the following strengthening of \cref{thm:LU_imply_LCr}.

\begin{corollary} \label{cor:LU_LCr}
    If two graphs of order at most $2^{r+3}-1$ are LU-equivalent, they are LC$_r$-equivalent. 
\end{corollary}

\begin{proof}
    Suppose that $G_1$ and $G_2$ of order $n \ls 2^{r+3}-1$ are LU-equivalent. According to \cref{thm:LU_imply_LCr}, $G_1$ and $G_2$ are LC$_{\lfloor n/2 \rfloor-1}$-equivalent. If $\lfloor n/2 \rfloor-1 \ls r$ then $G_1$ and $G_2$ are trivially LC$_r$-equivalent. Otherwise, according to \cref{cor:LCr_LCr-1}, $G_1$ and $G_2$ are LC$_r$-equivalent by direct induction. 
\end{proof}

\cref{cor:LU_LCr} provides a logarithmic bound on the level of generalised local complementations to consider for LU-equivalence: if two graphs of order $n>7$ are LU-equivalent then they are LC$_{\lceil \log_2(\frac{n+1}8)\rceil}$-equivalent. This bound leads to a quasi-polynomial time algorithm for LU-equivalence, as described in the next section. Notice that in \cref{sec:19qubits}, we elaborate on the consequences of \cref{cor:LU_LCr} on the minimal order of graphs that are LU- but not LC-equivalent. 

\subsection{An algorithm to recognise LU-equivalent graph states} \label{subsec:algorithm_lu}

According to \cref{thm:algolcr}, we have an algorithm that recognises two LC$_r$-equivalent graphs of order $n$ in runtime $O(r n^{r+2.38} + n^{6,38})$. According to \cref{cor:LU_LCr}, $G_1$ and $G_2$ are LU-equivalent if and only if they are LC$_r$-equivalent, where $r=\log_2(n)+O(1)$. Thus, our algorithm that decides LC$_r$-equivalence translates directly to an algorithm that decides LU-equivalence.

\begin{theorem}
    There exists an algorithm that decides if two graphs are LU-equivalent with runtime $n^{\log_2(n)+O(1)}$, where $n$ is the order of the graphs.
\end{theorem}

In comparison, Burchardt et al.~algorithm for LU-equivalence \cite{burchardt2024algorithmverifylocalequivalence} has two sources of exponential time complexity. The logarithmic upper bound on the level of generalised local complementation we introduce may mitigate one of these sources (making one parameter of the complexity quasi-polynomial), but does not affect a priori the second one, which is roughly speaking the number of connected components of an intersection graph related to the MLS cover. 

\section{LU- and LC-equivalence coincide for graph states up to 19 qubits} \label{sec:19qubits}

It is known that there exists a pair of 27-vertex graphs that are not LC-equivalent, but  LU-equivalent, more precisely they are LC$_2$-equivalent \cite{Ji07,Tsimakuridze17}. It is still  an open question whether this is a minimal example (in number of vertices). In other words, does a pair of graphs that are LU-equivalent but not LC-equivalent on 26 vertices or less exist? In theory, one could check every pair of graphs of order up to 26, but the rapid combinatorial explosion in the number of graphs as the number of vertices increases, makes it unfeasible in practice.

The best bound known so far\footnote{In \cite{burchardt2024algorithmverifylocalequivalence} it is proved that the number of LU- and LC-orbits of \textbf{unlabelled} graphs of order up to 11 is the same.} is that for graphs of order up to 8, LU=LC i.e.~LU- and LC-equivalence coincide  \cite{CABELLO20092219}. The results of \cref{subsec:bounds}  (see \cref{cor:LU_LCr}) already imply a substantial improvement on this bound: LU=LC for graphs of order up to 15. Furthermore, for graphs of order up to 31, LU=LC$_2$, i.e.~if two graphs of order up to 31 are LU-equivalent, they are LC$_2$-equivalent. Thus, asking whether LU=LC holds for graphs of order up to 26 is equivalent to asking whether LC$_2$=LC holds for graphs of order up to 26. One direction is to study when a 2-local complementation on an multiset $S$ can be implemented using only usual local complementations over vertices in the support of $S$. If this were to be the case for every graph of order up to 26, it would show that the 27-vertex counterexample is minimal in number of vertices. In the following we study the structure of 2-local complementation to prove that LU=LC holds for graph of order up to 19.

According to \cref{lemma:sizeZ}, if there are at most 4 vertices outside the support of some 2-incident independent multiset $S$, then a 2-local complementation on $S$ can be implemented by usual local complementations. In the peculiar case of 2-local complementation, we are able to use computer-assisted generation (see \cref{app:computer}) to extend the result. The code is available at \cite{codelulc19}.

\begin{restatable}{lemma}{computerassisted} \label{lemma:2lc6}
    Let $S$ be a $2$-incident independent multiset of a graph $G$. If $|V \sm \supp(S)|\ls 5$, or if $|V \sm \supp(S)|= 6$ and $|\supp(S)|\ls 20$, then a 2-local complementation on $S$ can be implemented by local complementations over a subset of $\supp(S)$. 
\end{restatable}

Likewise, according to \cref{lemma:exp_support} and \cref{prop:nontrivial}, if the support of some 2-incident independent multiset $S$ is of size at most 10, then a 2-local complementation on $S$ can be implemented by usual local complementations. To extend this result to 2-incident independent multisets whose supports is of size at most 12, we first study the case of twin-less sets (two distinct non-connected vertices $u$ and $v$ are twins if $N_G(u) = N_G(v)$).

\begin{restatable}{lemma}{lessthantwelve} \label{lemma:lessthan12}
    Let $S$ be a 2-incident independent set of a graph $G$ such that $S$ does not contain any twins and $|S| \ls 12$. Then, $G \star^2 S = G$.
\end{restatable}

The proof of \cref{lemma:lessthan12} is an induction over the number of vertices connected to $S$ and is given in \cref{app:lessthan12}.

According to \cref{prop:decomposition}, any 2-local complementation can be decomposed into 1- and 2-local complementations over sets. Furthermore, one can check that if an 2-incident independent set $S$ contains two twins $u$ and $v$, then a 2-local complementation over $S$ has the same effect as a 2-local complementation over $S\sm\{u,v\}$ followed by a local complementation over $u$. Thus, the action of a 2-local complementation can be described by a 2-local complementation over a twin-less set followed by usual local complementations. Then, \cref{lemma:lessthan12} can be applied on the twin-less set to yield the following result:

\begin{lemma} \label{lemma:2lc12}
    Let $S$ be a 2-incident independent multiset of a graph $G$ such that $|supp(S)|\ls 12$. Then, a 2-local complementation over $S$ can be implemented by local complementations over a subset of $supp(S)$.
\end{lemma}

According to \cref{lemma:2lc6} and \cref{lemma:2lc12}, if a 2-incident independent multiset $S$ satisfies $|supp(S)|\ls 12$ or $|V\sm\supp(S)|\ls 5$, or alternatively if $|supp(S)|\ls 20$ and $|V\sm\supp(S)|= 6$, then a 2-local complementation over $S$ can be implemented by usual local complementations. Thus, for graphs of order up to 19, any 2-local complementation can be implemented by usual local complementations, implying LU=LC. In other words, a 2-local complementation that cannot be implemented by usual local complementation is possible only on a graph of order at least 20. We summarise our findings in the following proposition:

\begin{proposition} 
    LU- and LC-equivalence coincide for graph states up to 19 qubits.
\end{proposition}

\section{Conclusion}

In this paper, we have introduced a quasi-polynomial runtime algorithm to recognise the LU-equivalence of graph states -- and more generally stabilizer states -- based on a recent generalisation of local complementation that captures the LU-equivalence of graph states. A key component of our approach is a new, nontrivial logarithmic bound on the level of the generalised local complementation.

We have also extended the well-known Bouchet algorithm to recognise the LC-equivalence of graph states, by allowing the addition of arbitrary linear constraints. This extension enables finer control over local complementations (or local Clifford operators) in the LC-equivalence problem, and we believe it will have broader applications.

We have also made significant progress in understanding the structure of quantum entanglement by demonstrating that LC-equivalence and LU-equivalence coincide for graph states with up to 19 qubits, extending the previously known bound of 8 qubits. 
The smallest known example of a pair of graph states that are LU- but not LC-equivalent consists of 27 qubits. A natural next step is to determine whether LU- and LC-equivalence continue to coincide for graph states up to 26 qubits or, alternatively, to find a counterexample in the range of 20 to 26 qubits. As shown in this work, leveraging generalised local complementation should facilitate this exploration.

\section*{Acknowledgements}
The authors want to thank Adam Burchardt for fruitful discussions. This work is supported by the the \emph{Plan France 2030} through the PEPR integrated project EPiQ ANR-22-PETQ-0007 and the HQI platform ANR-22-PNCQ-0002; and by the European projects Quantum Flagship NEASQC, European High-Performance Computing HPCQS and MSCA Staff Exchanges Qcomical HORIZON-MSCA-2023-SE-01. The project is also supported by the \emph{Maison du Quantique} MaQuEst. 


\bibliography{reflipics}

\appendix

\section{Alternative proof of Proposition \ref{prop:Bouchet}}
\label{app:Bouchet}

Recall that $G$ and $G'$ are said LC-equivalent if there exists a sequence of (possibly repeating) vertices $a_1, \cdots, a_m$ such that $G' = G \star a_1 \star \cdots \star a_m$.
According to Bouchet, $G$ and $G'$ are LC-equivalent if and only if there exist $A,B,C,D\subseteq V$ such that
\begin{itemize}
\item[(i)] $\forall u,v \in V$,\\
$|B\cap N_G(u)\cap N_{G'}(v)| + |A\cap N_G(u)\cap \{v\}| +  |D\cap \{u\}\cap N_{G'}(v)| + |C\cap \{u\}\cap \{v\}| = 0\bmod 2$
\item[(ii)]
$A\cap D~\Delta~ B\cap C = V$
\end{itemize}

We begin by proving by induction the "only if" part of the statement. First, notice that equations (i) and (ii) are satisfied when $G = G'$ with $A = D = V$ and $B = C =  \emptyset$. Indeed, let $u,v \in V$:
\begin{align*}
    & |B\cap N_G(u)\cap N_{G}(v)| + |A\cap N_G(u)\cap \{v\}| +  |D\cap \{u\}\cap N_{G}(v)| + |C\cap \{u\}\cap \{v\}|\\
    &= |N_G(u)\cap \{v\}| +  |\{u\}\cap N_{G}(v)|\\
    &= 0 \bmod 2
\end{align*}
Furthermore, $A\cap D~\Delta~ B\cap C = V$.

Now, suppose that $G$ and $G'$ are equivalent and there exist $A,B,C,D\subseteq V$  satisfying (i) and (ii). Applying a local complementation on some vertex $w$ in $G$ results in the graph $G \star w$, which is also LC-equivalent to $G'$. Define 
\begin{itemize}
    \item $A' = A \Delta ( \{w\}\cap C)$
    \item $B' = B \Delta ( \{w\}\cap D)$
    \item $C' = C \Delta ( N_{G}(w)\cap A)$
    \item $D' = D \Delta ( N_{G}(w)\cap B)$
\end{itemize}
Let us show that $A',B',C',D'\subseteq V$  satisfy (i) and (ii) for the graphs $G \star u$ and $G'$.

\noindent{\bf Proof that $A',B',C',D'$  satisfy (i).}

Let $u,v \in V$.
\begin{align*}
    & |B'\cap N_{G \star w}(u)\cap N_{G'}(v)| + |A'\cap N_{G \star w}(u)\cap \{v\}| +  |D'\cap \{u\}\cap N_{G'}(v)| + |C'\cap \{u\}\cap \{v\}|\\
    &=|(B \Delta ( \{w\}\cap D))\cap N_{G \star w}(u)\cap N_{G'}(v)| + |(A \Delta ( \{w\}\cap C))\cap N_{G \star w}(u)\cap \{v\}|\\
    &~~+ |(D \Delta ( N_{G}(w)\cap B))\cap \{u\}\cap N_{G'}(v)| + |(C \Delta ( N_{G}(w)\cap A))\cap \{u\}\cap \{v\}|
\end{align*}
If $u\not\sim_{G} w$, then $N_{G \star w}(u) = N_{G}(u)$ and $N_G(w)\cup \{u\} = N_G(u)\cup \{w\} = \emptyset$:
\begin{align*}
    &= |B\cap N_G(u)\cap N_{G'}(v)| + |A\cap N_G(u)\cap \{v\}| +  |D\cap \{u\}\cap N_{G'}(v)| + |C\cap \{u\}\cap \{v\}|\\
    &~~+ | \{w\}\cap D\cap N_G(u)\cap N_{G'}(v)| + |\{w\}\cap C\cap N_G(u)\cap \{v\}|\\
    &~~+ |N_{G}(w)\cap B\cap \{u\}\cap N_{G'}(v)| + |N_{G}(w)\cap A \cap \{u\}\cap \{v\}| \bmod 2\\
    &= 0 \bmod 2
\end{align*}
If $u\sim_{G} w$, then $N_{G \star w}(u) = N_{G}(u) \Delta N_{G}(w) \Delta \{u\}$, $N_G(w)\cup \{u\} = \{u\}$ and $N_G(u)\cup \{w\} = \{w\}$, thus :
\begin{align*}
    &=|(B \Delta ( \{w\}\cap D))\cap (N_{G}(u) \Delta N_{G}(w) \Delta \{u\})\cap N_{G'}(v)|\\
    &~~+ |(A \Delta ( \{w\}\cap C))\cap (N_{G}(u) \Delta N_{G}(w) \Delta \{u\})\cap \{v\}|\\
    &~~+ |(D \Delta ( N_{G}(w)\cap B))\cap \{u\}\cap N_{G'}(v)| + |(C \Delta ( N_{G}(w)\cap A))\cap \{u\}\cap \{v\}|\\
    &= |B\cap N_G(u)\cap N_{G'}(v)| + |A\cap N_G(u)\cap \{v\}| +  |D\cap \{u\}\cap N_{G'}(v)| + |C\cap \{u\}\cap \{v\}|\\
    &~~+|B\cap N_G(w)\cap N_{G'}(v)| + |A\cap N_G(w)\cap \{v\}| +  |D\cap \{w\}\cap N_{G'}(v)| + |C\cap \{w\}\cap \{v\}|\\
    &~~+|B\cap\{u\}\cap N_{G'}(v)|+|A\cap\{u\}\cap\{v\}|+|B\cap\{u\}\cap N_{G'}(v)|+|A\cap\{u\}\cap \{v\}|\\
    &~~+ |\{w\} \cap D \cap N_G(w) \cap N_{G'}(v)| + |\{w\} \cap C \cap N_G(w) \cap \{v\}|\\
    &~~+ |\{w\} \cap D \cap \{u\} \cap N_{G'}(v)| + |\{w\} \cap C \cap \{u\}\cap \{v\}|\\
    &= 0 \bmod 2
\end{align*}

\noindent{\bf Proof that $A',B',C',D'$  satisfy (ii).}
\begin{align*}
    & (A'\cap D')\Delta(B'\cap C')\\
    &= \left((A \Delta ( \{w\}\cap C))\cap (D \Delta ( N_{G}(w)\cap B))\right)\Delta\left((B \Delta ( \{w\}\cap D))\cap (C \Delta ( N_{G}(w)\cap A))\right)\\
    &= (A\cap D)\Delta(A\cap N_{G}(w)\cap B)\Delta(D\cap \{w\}\cap C)\Delta(B\cap C)\Delta(B \cap  N_{G}(w)\cap A)\Delta(C \cap\{w\}\cap D)\\
    &=(A\cap D)\Delta(B\cap C) = V
\end{align*}

Now we prove the "if" part of the statement. The proof is very similar to a proof in \cite{claudet2024local} regarding the relation between local complementation and local Clifford operators. Let $G$ and $G'$ be two graphs defined on the same vertex set $V$ along with $A, B, C, D$ satisfying (i) and (ii). Condition (ii) implies that for some vertex $u \in V$, 6 cases can occur:
\begin{enumerate}
    \item $u \in A \cap \overline B \cap \overline C \cap D$; 
    \item $u \in A \cap B \cap \overline C \cap D$; 
    \item $u \in A \cap \overline B \cap C \cap D$; 
    \item $u \in \overline A \cap B \cap C \cap \overline D$; 
    \item $u \in A \cap B \cap C \cap \overline D$; 
    \item $u \in \overline A \cap B \cap C \cap D$. 
\end{enumerate}
We call $V_1$ (resp. $V_2$, $V_3$, $V_4$, $V_5$, $V_6$) the set of vertices in case 1 (resp. 2, 3, 4, 5, 6). Notice that $V = V_1$ implies $G = G'$, indeed condition (ii) implies that for any $u,v \in V$, $|N_G(u)\cap \{v\}| +  |\{u\}\cap N_{G'}(v)| = 0 \bmod 2$ i.e.~$u \sim_G v \Leftrightarrow u \sim_{G'} v$. Furthermore, applying a local complementation on a vertex $w$ of $G$ changes the sets $A$, $B$, $C$, $D$, thus it changes in which case a vertex is. The changes are given in the following table (the case in which unwritten vertices are remain unchanged). 

\begin{center}
    \begin{tabular}{|c|c|}
    \hline
    \multicolumn{2}{|c|} {Case of $w$ in}\\
    $~~G~~$& $G\star w$\\
    \hline
    1&2\\
    \hline
    2&1\\
    \hline
    3&6\\
    \hline
    4&5\\
    \hline
    5&4\\
    \hline
    6&3\\
    \hline
    \end{tabular}\qquad\begin{tabular}{|c|c|}
    \hline
    \multicolumn{2}{|c|}{Case of $u {\in}N_G(w)$ in}\\
    $~~~G~~~$& $G\star w$\\
    \hline
    1&3\\
    \hline
    2&5\\
    \hline
    3&1\\
    \hline
    4&6\\
    \hline
    5&2\\
    \hline
    6&4\\
    \hline
    \end{tabular}
    \qquad \begin{tabular}{|c|c|}
    \hline
    \multicolumn{2}{|c|}{Case of $w_1$ (or $w_2$) in}\\
    $~~G~~$&$G\wedge w_1 w_2$\\
    \hline
    1&4\\
    \hline
    2&6\\
    \hline
    3&5\\
    \hline
    4&1\\
    \hline
    5&3\\
    \hline
    6&2\\
    \hline
    \end{tabular}
\end{center}

The table indicates that if $G$ and $G'$ are LC-equivalent and $A, B, C, D$ satisfy (i) and (ii), then, for $G \wedge w_1 w_2$ and $G'$, $A', B', C', D'$ satisfy (i) and (ii) where:
\begin{itemize}
    \item $A' = (A \sm \{w_1,w_2\}) \cup (C \cap \{w_1,w_2\})$
    \item $B' = (B \sm \{w_1,w_2\}) \cup (D \cap \{w_1,w_2\})$
    \item $C' = (C \sm \{w_1,w_2\}) \cup (A \cap \{w_1,w_2\})$
    \item $D' = (D \sm \{w_1,w_2\}) \cup (B \cap \{w_1,w_2\})$
\end{itemize}

Let us design an algorithm that produces a sequence of (possibly repeating) vertices $s = (a_1, \cdots, a_m)$ such that $G' = G \star a_1 \star \cdots \star a_m$. Initialise $G_0 = G$, $s_0 = [~]$ an empty sequence of vertices and $A_0 = A$, $B_0 = B$, $C_0 = C$, $D_0 = D$.

\begin{enumerate}
    \item If there is a vertex $u$ in case 2 or 6: let $s_0 \leftarrow s_0 + [u]$, $G_0 \leftarrow G_0 \star u$, $A_0 \leftarrow A_0 \Delta ( \{u\}\cap C_0)$, $B_0 \leftarrow B_0 \Delta ( \{u\}\cap D_0)$, $C_0 \leftarrow C_0 \Delta ( N_{G_0}(u)\cap A_0)$, $D_0 \leftarrow D_0 \Delta ( N_{G_0}(u)\cap B_0)$. Repeat until there is no vertex in case 2 or 6 left. 
    \item If there is a vertex $u$ in case 4 or 5: let $v \in N_{G_0}(u)$ such that $v$ is also in case 4 or 5. 
    Let $s_0 \leftarrow s_0 + [u, v, u]$, $G_0 \leftarrow G_0 \wedge u v$, $A_0 \leftarrow (A_0 \sm \{u,v\}) \cup (C_0 \cap \{u,v\})$, $B_0 \leftarrow (B_0 \sm \{u,v\}) \cup (D_0 \cap \{u,v\})$, $C_0 \leftarrow (C_0 \sm \{u,v\}) \cup (A_0 \cap \{u,v\})$, $D_0 \leftarrow (D_0 \sm \{u,v\}) \cup (B_0 \cap \{u,v\})$. Then go to step 1.
\end{enumerate}

\noindent{\bf Correctness.~} The evolution of $A_0$, $B_0$, $C_0$ and $D_0$ at each iteration of the algorithm ensures that (i) and (ii) are satisfied for $G \star s_0$ and $G'$. In step 2, if there is a vertex u in case 4 or 5, let us show that there exists $v \in N_{G_0}(u)$ such that $v$ is also in case 4 or 5. Notice that in step 2, no vertex is in case 2 or 6. Suppose by contradiction that every vertex in $N_{G_0}(u)$ is in case 1 or 3, i.e.~for every $v \in N_{G_0}(u)$, $v \in A \cap \overline B  \cap D$. Then $|B_0\cap N_{G_0}(u)\cap N_{G'}(u)| + |A_0\cap N_{G_0}(u)\cap \{u\}| +  |D_0\cap \{u\}\cap N_{G'}(u)| + |C_0\cap \{u\}\cap \{u\}| = |C_0 \cap \{u\}| = 1 \bmod 2$, contradicting (ii). At the end of the algorithm, every vertex is in case 1 or 3. Actually, every vertex is in case 1. Suppose by contradiction there is a vertex u in case 3. Then $|B_0\cap N_{G_0}(u)\cap N_{G'}(u)| + |A_0\cap N_{G_0}(u)\cap \{u\}| +  |D_0\cap \{u\}\cap N_{G'}(u)| + |C_0\cap \{u\}\cap \{u\}|  = |C_0 \cap \{u\}| = 1 \bmod 2$, contradicting (ii). Thus, at the end of the algorithm, $G' = G \star s_0$.

\noindent{\bf Termination.} The number of vertices in case 1 or 3 strictly increases at each iteration of the algorithm.

\section{Interpretation of the constraints in terms of local Clifford operators} \label{app:interpretationClifford}

There is a one-to-one correspondence between the solutions to  equations $(i)$ and $(ii)$ and the local Clifford operators (up to Pauli operators) that maps $\ket G$ to $\ket {G'}$. In particular if $A,B,C,D$ satisfy equations $(i)$ and $(ii)$, then $\ket{G'} = e^{i\theta}\bigotimes_{v\in V}U_v \ket{G}$ where for any $v \in V$, $U_v$ is equal,  up to a Pauli operator, to:

\begin{table}[h]
    \centerline{\begin{tabular}{clccl}
     $I$ & if  $v\in \overline{B}\cap \overline{C}$&$\qquad\qquad$& $H$ & if $v\in \overline{A}\cap \overline{D}$\\
      $Z(\pi/2)$ &if $v\in \overline{B}\cap C$&& $Z(\pi/2)H$ & if $v\in \overline{A}\cap {D}$\\
      $X(\pi/2)$ & if $v\in B \cap \overline{C}$&&$X(\pi/2)H$& if $v\in {A}\cap \overline{D}$
    \end{tabular}}
    \caption{Corresponding Clifford operators\label{table:encoding}.}
\end{table}

More details on the LC-equivalence of graphs and the corresponding Clifford operators can be found in \cite{VdnEfficientLC,Hein06}. From a graph state point of view, \cref{prop:extendedBouchet} provides an efficient algorithm to decide whether two graph states are LC-equivalent under some constraints on the Clifford operators. Notice that such constraints should be expressible as a linear equation through the correspondence given in \cref{table:encoding}. We give below a non-exhaustive family of constraints expressible as linear equations (in the following, $k$ denotes an integer).

\begin{itemize}
    \item $v \notin B$:~ $U_v$ is $Z(k\pi/2)$ up to Pauli;
    \item $v \notin C$:~ $U_v$ is $X(k\pi/2)$ up to Pauli;
    \item $v \notin A$:~ $U_v$ is $Z(k\pi/2) H$ up to Pauli;
    \item $v \notin D$:~ $U_v$ is $X(k\pi/2) H$ up to Pauli;
    \item $v \notin \overline{B}\cap \overline{C}$:~ $U_v$ is a Pauli;
    \item $v \notin \overline{A}\cap \overline{D}$:~ $U_v$ is $H$ up to a Pauli;
    \item $v \in A$ iff $v \in D$:~ $U_v$ is I, $X(\pi/2)$, $Z(\pi/2)$ or $H$ up to Pauli, i.e.~$U_v^2$ is a Pauli;
    \item $v \in A$ iff $v \in B$:~ $U_v$ is $X(\pi/2)$ or $X(k\pi/2) H$ up to Pauli;
    \item $v \in A$ iff $w \in A$, $v \in B$ iff $w \in B$, $v \in C$ iff $w \in C$, $v \in D$ iff $w \in D$:~ $U_v=U_w$ up to Pauli.
\end{itemize}

\section{Proof of Lemma \ref{lemma:standardform}} \label{app:standardform}

\standardform*

\begin{proof}
    To prove the proposition, we introduce an algorithm that transforms the input graphs $G_1, G_2$ into graphs in standard form with respect to the same MLS cover by means of local complementations. The notation $\wedge $ refers to the pivoting operation: $G\wedge uv \defeq G \star u \star v = G \star v \star u$. The dimension of a minimal local set refers to the logarithm in base 2 of its number of generators, in particular a minimal local set of dimension 2 induces vertices of type $\bot$, and minimal local sets can be of dimension either 1 or 2 (see \cite{claudet2024local}). The action of the local complementation and the pivoting on the type of the vertices is given in the following table (the types of the unwritten vertices remain unchanged). 
    \begin{center}
    \begin{tabular}{|c|c|}
    \hline
    \multicolumn{2}{|c|} {Type of $u$ in}\\
    $~~G~~$& $G\star u$\\
    \hline
    X&X\\
    \hline
    Y&Z\\
    \hline
    Z&Y\\
    \hline
    $\bot$&$\bot$\\
    \hline
    \end{tabular}\qquad\begin{tabular}{|c|c|}
    \hline
    \multicolumn{2}{|c|}{Type of $v {\in}N_G(u)$ in}\\
    $~~~G~~~$& $G\star u$\\
    \hline
    X&Y\\
    \hline
    Y&X\\
    \hline
    Z&Z\\
    \hline
    $\bot$&$\bot$\\
    \hline
    \end{tabular}
    \qquad \begin{tabular}{|c|c|}
    \hline
    \multicolumn{2}{|c|}{Type of $u$ (or $v$) in}\\
    $~~G~~$&$G\wedge uv$\\
    \hline
    X&Z\\
    \hline
    Y&Y\\
    \hline
    Z&X\\
    \hline
    $\bot$&$\bot$\\
    \hline
    \end{tabular}
    \end{center}

    The algorithm reads as follows:

    \begin{enumerate}
        \item Compute an MLS cover $\mathcal M$ of $G_1$ \cite{claudet2024covering}. If $\mathcal M$ is not an MLS cover of $G_2$, then $G_1$ and $G_2$ are not LU-equivalent.
        \item If there is an XX-edge (i.e.~an edge $uv$ s.t. both $u$ and $v$ are of type X with respect to $\mathcal M$) in $G_1$ or $G_2$: apply a pivoting on it.\\
        Repeat until there is no XX-edge left.
        \item If there is a XY-edge in $G_1$ or $G_2$: apply a local complementation on the vertex of type X, then go to step 2.
        \item If there is a vertex of type Y in $G_1$ or $G_2$: apply a local complementation on it, then go to step 2.
        \item If there is an X$\bot$-edge in $G_1$ or $G_2$: apply a pivoting on it.\\
        Repeat until there is no X$\bot$-edge left.

        \item If there is an XZ-edge $uv$ in $G_1$ or $G_2$ such that $v\prec u$: apply a pivoting on $u v$.\\
        Repeat until for every XZ-edge $uv$ in $G_1$ or $G_2$, $u\prec v$.
 
        \item If there is a vertex $u$ of type X in $G_1$ (resp. $G_2$) such that $\{u\} \cup N_{G_1}(u)$ (resp.  $\{u\} \cup  N_{G_2}(u)$) is not a minimal local set of dimension 1: find a minimal local set $M$ contained in $\{u\} \cup N_{G_1}(u)$ (resp.  $\{u\} \cup  N_{G_2}(u)$) and check that $M$ is a minimal local set of same dimension in both graphs (if not, they are not LU-equivalent). If this is the case, add $M$ to $\mathcal M$ then go to step 5.

        \item For every vertex $u$ of type X in $G_1$ (resp. $G_2$), add $\{u\} \cup N_{G_1}(u)$ (resp.  $\{u\} \cup N_{G_2}(u)$) to $\mathcal M$.
    \end{enumerate}

    \noindent{\bf Correctness.~} When step 2 is completed, there is no XX-edge. Step 3 transforms the neighbours of type Y into vertices of type Z. No vertex of type Y is created as there is no XX-edge before the local complementation. When step 3 is completed, there is no XX-edge nor any XY-edge. Step 4 transforms the vertex of type Y into a vertex of type Z. No vertex of type Y is created as there is no XY-edge before the local complementation. When step 4 is completed, there is no vertex of type Y nor any XX-edge. In step 5, applying a pivoting on an X$\bot$-edge transforms the vertex of type X into a vertex of type Z. No XX-edge is created, as the vertex of type X has no neighbour of type X before the pivoting. When step 5 is completed, there is no vertex of type Y and each neighbour of a vertex of type X is of type Z. In step 6, applying a pivoting on an XZ-edge permutes the type of the two vertices, and preserves the fact that each neighbour of a vertex of type X is of type Z. In step 7, adding a minimal local set to  $\mathcal M$ may only change the type of some vertices to $\bot$. When step 7 is completed, for every vertex $u$ of type X in $G_1$ (resp. $G_2$), $\{u\} \cup N_{G_1}(u)$ (resp.  $\{u\} \cup N_{G_2}(u)$) is a minimal local set. Thus, in step 8, adding those minimal local set leave the types invariant. When step 8 is completed, $G_1$ and $G_2$ are in standard form with respect to $\mathcal M$.
    
    \noindent{\bf Termination.} The quantity $2|V^{G_1}_Y|+|V^{G_1}_X| + 2|V^{G_2}_Y|+|V^{G_2}_X|$, where $|V^{G_i}_X|$ (resp. $|V^{G_i}_Y|$) denotes the number of vertices of type X (resp. Y) with respect to $\mathcal M$ in $G_i$, strictly decreases at steps 2 to 5, which guarantees to reach step 6. At step 6, $V^{G_i}_X$ is updated as follows: exactly one vertex $u$ is removed from the set and is replaced by a vertex $v$ such that $v\prec u$, which guarantees the termination of step 6. Each time a minimal local set is added to $\mathcal M$ in step 7, at least one vertex not of type $\bot$ becomes of type $\bot$. Indeed, without loss of generality, let $M$ be a minimal local set in $\{u\} \cup N_{G_i}(u)$, assuming $\{u\} \cup N_{G_i}(u)$ is not a minimal local set of dimension 1 generated by $\{u\}$. If $M$ is of dimension 2, every vertex of $M$ becomes of type $\bot$ when adding $M$ to $\mathcal M$. Else, $M$ is generated by a set containing at least one vertex of type Z, which becomes of type $\bot$ when adding $M$ to $\mathcal M$.

    \noindent{\bf Complexity.} The time-complexity of the algorithm is given by the time-complexity of step 1, as it is asymptotically the most computationally expensive step. The MLS cover can be computed in runtime $O(n^{6.38})$, where $n$ is the order of the graph. It should be noted that giving the type of a vertex with respect to some minimal local set can be done in time-complexity $O(n^3)$. Indeed, in the case of a minimal local $L$ set of dimension 1, finding $D$ such that $L = D \cup Odd_G(D)$ reduces to finding the kernel of some matrix with coefficients in $\mathbb F_2$ (see details in \cite{claudet2024covering}), which can be done using Gaussian elimination.
\end{proof}

\begin{remark}
    In the main algorithm described in \cref{subsec:algorithm_lcr}, the number of vertices of type $Z$ is directly linked to the runtime of the algorithm. Thus, it is preferable that the number of vertices of type Z is low. 
    It should be noted that at any point of the algorithm, following \cite{claudet2024local}, if there exists a set $K$ of vertices of type Z of size more than $\lfloor n/2 \rfloor +1$, then a minimal local set $L$ within $K$ can be found, and adding $L$ to $\mathcal M$ transforms at least one vertex of type $Z$ into a vertex of type $\bot$. Repeating the operation leads to a graph where the number of vertices of type $Z$ is at most $\lfloor n/2 \rfloor$. The algorithm can then be restarted from step 2. Notice that the number of vertices of type $\bot$ never decreases, thus this procedure will be performed less than $n$ times. Adding random minimal local sets to $\mathcal M$ is another way of trying to reduce the number of vertices of type Z.
\end{remark}

\section{Proof of Lemma \ref{lemma:new_graphs}} \label{app:new_graphs}

\newgraphs*

\begin{proof}
    Suppose $G_1$ and $G_2$ LC$_r$-equivalent. According to \cref{lemma:standardform_LCr_lc}, $G_1$ and $G_2$ are related by a single $r$-local complementation over a multiset $S$ whose support lies in $V_X$, along with a sequence of local complementations on the vertices of $V \sm (V_X \cup V_Z)$. We note $G_2 = G_1 \star^r S \star u_1 \star \cdots u_k$. The edges toggled by an $r$-local complementation over $S$ are described by an element $\omega \in \Omega$. Let us decompose $\omega$ as a linear combination of basis vectors of $\Omega$: $\omega = \omega_1 + \cdots + \omega_t$, where each $\omega_i \in \mathcal B$. In $G^{\#}_1$ and $G^{\#}_2$, let $V_\omega = \bigcup_{i \in [1,t]} \mathcal P^{\omega_i}$.
    Then, $G^{\#}_2 = G^{\#}_1  \star^1 V_\omega \star u_1 \star \cdots u_k$. Note that the 1-local complementation over the set $V_\omega$ corresponds to the composition of local complementations on each element of $V_\omega$. Thus, $G_1$ is mapped to $G_2$ by a sequence of local complementations that satisfy the additional constraints.

    Conversely, suppose there exists a sequence of (possibly repeating) vertices $a_1, \cdots, a_m$ such that $G^{\#}_2 = G^{\#}_1 \star a_1 \star \cdots \star a_m$ satisfying the additional constraints. As local complementations on new vertices commute with each other and with local complementations on vertices of $V \sm (V_X \cup V_Z)$, one can take apart the vertices of the sequence among the new vertices, so there exists a set $V_0$ of new vertices and vertices $u_i$ in $V \sm (V_X \cup V_Z)$ such that $G^{\#}_2 = G^{\#}_1  \star^1 V_0 \star u_1 \star \cdots u_k$. The additional constraints imply that $V_0$ is an union of sets of vertices corresponding respectively to some elements $\omega_i \in \mathcal B$. Let $\omega \in \Omega$ be the sum of these elements. By construction, there exists a multiset $S$ in the original graphs whose action is described by $\omega$, implying $G_2 = G_1 \star^r S \star u_1 \star \cdots u_k$. Thus, $G_1$ and $G_2$ are LC$_r$-equivalent.
\end{proof}

\section{Proof of Lemma \ref{lemma:sizeZ}} \label{app:sizeZ}

\sizeZ*

\begin{proof}

    According to \cref{prop:nontrivial}, $S$ is genuine. Let $K_\text{max}$ be one of the biggest subset (by inclusion) of $V \sm \supp(S)$ such that $\sum_{N_{G}(u)=K_\text{max}}S(u)$ is odd: then $S \bullet\Lambda_G^{K_\text{max}}$ is odd. If $|V \sm \supp(S)| \ls r+1$, then  $|K_{max}|\ls r+1$, contradicting the $r$-incidence of $S$. Thus $|V\sm\supp(S)| \gs |K_{max}| \gs r+2$. 

    Now, suppose $|V\sm\supp(S)| = r+2$, i.e.~$K_{max}=V\sm\supp(S)$. Let us prove that for any $K \se V \sm \supp(S)$ s.t. $|K|>1$, $\sum_{N_{G}(u)=K}S(u)$ is odd, by induction over the size of $K$. First, notice that $\sum_{N_{G}(u)=V \sm \supp(S)}S(u)$ is odd.
    Then, let $K_0 \subsetneq V \sm \supp(S)$ s.t. $|K_0|>1$.
    \begin{align*}
        &S \bullet\Lambda_G^{K_0} =\!\!\!\! \sum_{K_0\se K\se V \sm \supp(S)}\sum_{N_{G}(u)=K}S(u)
        = \sum_{N_{G}(u)=K_0}S(u) + \sum_{K_0\subsetneq K\se V \sm \supp(S)}\sum_{N_{G}(u)=K}S(u)\\
        &= \sum_{N_{G}(u)=K_0}S(u) + |\{K\se V \sm \supp(S)~|~K_0 \subsetneq K\}| \bmod 2 \text{~~by hypothesis of induction}\\
        &= \sum_{N_{G}(u)=K_0}S(u) + 1 \bmod 2 
    \end{align*}
    Thus, $\sum_{N_{G}(u)=K_0}S(u)$ is odd, as $S \bullet\Lambda_G^{K_0}$ is even by $r$-incidence of $S$. As a consequence, for any $K \se V \sm \supp(S)$ s.t. $|K|>1$, there exists at least one vertex $u\in \supp(S)$ s.t. $N_{G}(u)=K$. 

    Let $S'$ be the multiset obtained from $S$ by choosing, for every $K \se V \sm \supp(S)$ s.t. $|K|>1$, a single vertex $u \in \supp(S)$ s.t. $N_{G}(u)=K$, and setting $S'(u)=S(u)-1$ (the multiplicity of other vertices remain unchanged). Let us prove that $S'$ is $r$-incident and $G \star^r S = G\star^{r} S'$. First, $S'$ is $r$-incident. Indeed, let an integer $k\in [0,r)$, let $K_1\subseteq V\setminus \supp(S')$ be a set of size $k+2$, and let $k'=k-|K_1\cap \supp(S)|$. $S\bullet \Lambda_G^{K_1}$ is a multiple of $2^{r-k'-\delta(k')}$ by $r$-incidence of $S$, so is $S'\bullet \Lambda_G^{K_1}$, as $S'\bullet \Lambda_G^{K_1} = S\bullet \Lambda_G^{K_1} -|\{K \se V \sm \supp(S)~|~K_1\sm \supp(S)\se K\}| = S\bullet \Lambda_G^{K_1} -2^{r-k'}$. Then, if $u$ or $v$ is in $\supp(S)$, $u\sim_{G\star^r S} v ~\Leftrightarrow~ u\sim_{G\star^{r} S'} v ~\Leftrightarrow~ u\sim_{G} v$. If $u,v \in V \sm \supp(S)$:
    \begin{align*}
        u\sim_{G\star^r S} v &~\Leftrightarrow~\left(u\sim_{G} v ~~\oplus~~ S \bullet\Lambda_G^{u,v} = 2^{r-1}\bmod 2^{r}\right)\\
        &~\Leftrightarrow~\left(u\sim_{G} v ~~\oplus~~ S'\bullet \Lambda_G^{u,v} +2^{r} = 2^{r-1}\bmod 2^{r}\right)\\
        &~\Leftrightarrow~\left(u\sim_{G} v ~~\oplus~~ S'\bullet \Lambda_G^{u,v} = 2^{r-1}\bmod 2^{r}\right)
        ~\Leftrightarrow~ u\sim_{G\star^{r}S'} v
    \end{align*}
    Thus, $G\star^r S = G\star^r S'$. Notice also that $S'$ is not genuine. Thus, by \cref{prop:nontrivial} there exists an $S''$ such that $G\star^{r-1} S'' = G \star^r S' = G \star^r S$.
\end{proof}

\section{Computer-assisted study of 2-local complementation}
\label{app:computer}

We use the following lemma to drastically decrease the size of the space to explore when studying 2-local complementation.

\begin{lemma} \label{lemma:lifting}
    Let $S$ be a 2-incident independent multiset of a graph $G=(V,E)$ and suppose that there exists no set $A \se \supp(S)$ such that $G \star^2 S = G \star^1 A$. Then there exists a graph $G'=(V',E')$ bipartite with respect to a bipartition $S', V' \sm S'$ of the vertices such that:
    \begin{itemize}
        \item $S'$ is 2-incident;
        \item $S'$ contains no twins;
        \item $S'$ contains no vertex of degree 0 or 1;
        \item $|S'| \ls |\supp(S)|$;
        \item $|V' \sm S'|\ls |V \sm \supp(S)|$;
        \item there exists no set $A \se S'$ such that $G' \star^2 S' = G' \star^1 A$.
    \end{itemize}
\end{lemma}

\begin{proof}
    The proof is constructive, in the sense that we construct $G'$ and $S'$ from $G$ and $S$. According to \cref{prop:decomposition}, there exists set $S_1, S_2 \se \supp(S)$ such that $S_2$ is 2-incident and $G \star^2 S = G \star^2 S_2 \star^1 S_1$. Also, the existence of a set $A \se S_2$ such that $G \star^2 S_2 = G \star^1 A$ implies $G \star^2 S = G \star^1 (A \Delta S_1)$. Let $G' = G$ and $S' = S_2$, then apply the following operations:
    \begin{enumerate}
        \item Remove the edges between vertices of $V' \sm S'$;
        \item Remove each vertex of $S'$ of degree 0 or 1;
        \item If there exists a pair of twins $u,v$ in $S'$, remove $u$ and $v$. Repeat until $S'$ contains no twins.
    \end{enumerate}
    Notice that each operation preserves the 2-incidence of $S'$ and that  $G' \star^2 S' = G' \star^1 A$ for no set $A \se S'$.
\end{proof}

Given an integer $k$, let $\mathcal G_k$ be the class of graphs that are bipartite with respect to a bipartition $S, V \sm S$ of the vertices such that:
\begin{itemize}
    \item $S$ is 2-incident;
    \item $S$ contains no twins;
    \item $S$ contains no vertex of degree 0 or 1;
    \item $|V \sm S| = k$.
\end{itemize}
It is easy to generate each graph of $\mathcal G_k$, although the number of elements in $\mathcal G_k$ grows double exponentially fast with $k$. $S$ can be defined as a list of words in $\{0,1\}^k$ of weight at least 2. More precisely each vertex of $S$ is uniquely associated with a set of $V \sm S$ of size at least 2, its neighbourhood. Furthermore, the 2-incidence of $S$ implies that $S$ is uniquely determined by the set of its vertices of degree at least 4. Indeed, starting from a set containing only vertices of degree at least 4, the conditions of the form "$S\bullet \Lambda_G^K = 0 \bmod 2^{r-k-\delta(k)}$" translate into a procedure to find which vertices of degree 3 then 2 need to be added to the set so that $S$ is 2-incident. This proves that there is a bijection between $\mathcal G_k$ and lists of words in $\{0,1\}^k$ of weight at least 4. Thus, the size of $\mathcal G_k$ is exactly given by the formula $$ |\mathcal G_k| = 2^{\binom{k}{4} + \binom{k}{5} + \cdots + \binom{k}{k}}$$
For $k=1,2,3,4,5$ and $6$, the size of $\mathcal G_k$ is respectively $1,1,1,2,2^6 = 64$, and $2^{22} \sim 4 \times 10^6$ which is suitable for computation. But, even for $k$ as low as seven, the size of $\mathcal G_7$ is $2^{64} \sim 2 \times 10^{19}$.

For every $k$ from 1 to 6, we generate each graph $G$ of $\mathcal G_k$, along with the set $S$ defined above. Notice that a local complementation over a vertex $u$ of $S$ toggles the connectivity of some pairs of vertices of $V \sm S$, here the pairs where each end is a neighbour of $u$. In other words, to each vertex $u$ of $S$ we associate a vector in $\mathbb F_2^{\binom{k}{2}}$ corresponding to the action of the local complementation of $u$ on the graph. The set of the vectors corresponding to each vertex of $S$ spans a $\mathbb F_2$-vector space $\mathcal L$ describing the action of local complementation over vertices of $S$ on the graph. Using Gaussian elimination, we are able to compute a basis of $\mathcal L$. Furthermore, we compute the vector $x$ in $\mathbb F_2^{\binom{k}{2}}$ corresponding to the action of a 2-local complementation over $S$ on the graph. Checking if the action of a 2-local complementation over $S$ can be implemented by local complementations on vertices of $S$ amounts to checking if $x$ belongs to the vector space $\mathcal L$, which can be done efficiently using Gaussian elimination.

For $k \in [1,3]$, the set $S$ corresponding to the only graph $G \in \mathcal G_k$ is empty, hence a 2-local complementation over $S$ leaves $G$ invariant, i.e.~$G \star^2 S = G$. For $k = 4$, $S$ is either empty of contains 11 vertices; in both case it is easy to check that $G \star^2 S = G$. For $k = 5$, the computation shows that for each graph $G \in \mathcal G_5$, $G \star^2 S = G$. Now, fix $k=6$. For each graph $G \in \mathcal G_6$ such that the corresponding set $S$ contains at most 16 vertices, $G \star^2 S = G$. For each graph $G \in \mathcal G_6$ such that the corresponding set $S$ contains at most 20 vertices, a 2-local complementation over the corresponding set $S$ can be implemented by local complementations over vertices of $S$, i.e.~there exists a set $A \se S$ such that $G \star^2 S = G \star^1 A$. The property does not hold if $S$ is of size 21, for instance we recover the well-known 27-vertex counterexample to the LU=LC conjecture described in \cite{Tsimakuridze17}.

According to \cref{lemma:lifting}, this is enough to prove \cref{lemma:2lc6}:

\computerassisted*

\section{Proof of Lemma \ref{lemma:lessthan12}} \label{app:lessthan12}

\lessthantwelve*

\begin{proof}
    Given a set $S$, we define the strict neighbourhood of $S$ in the graph $G$ as the set of vertices not in $S$ that are connected to at least one vertex in $S$: $\delta_G(S) = \{v \in V \sm S ~|~ N_G(v) \cap S \neq \emptyset\}$. Let us prove, by induction on $|\delta_G(S)|$, that for any 2-incident independent set $S$ such that $S$ does not contain any twins and $|S| \ls 12$, $G \star^2 S = G$.

    First, by brute force, we check that the property is true whenever $|\delta(S)| \ls 6$ (see details in \cref{app:computer}, code available at \cite{codelulc19}). 

    Now, suppose the property true for an integer $t \gs 6$. Let $S$ be a 2-incident independent set such that $S$ does not contain any twins, $|S| \ls 12$, and $|\delta_G(S)| = t+1$. Let us prove that $G \star^{2} S = G$. Suppose by contradiction that $G \star^{2} S \neq G$. Then, there exists $u,v \in \delta(S)$ such that the edge between $u$ and $v$ is toggled by the 2-local complementation over $S$. 
    
    Let $w \in \delta(S) \sm \{u,v\}$, and let $G' = G[V \sm \{w\}]$, i.e.~$G'$ is the graph obtained from $G$ by removing the vertex $w$. $S$ is still a 2-incident independent set in $G'$ (but now it may contain twins). Notice that $\delta_{G'}(S) = \delta_{G}(S) -1$. The case where $S$ contains no twins contradicts the induction hypothesis. Suppose that $S$ does contain two twins, say $x$ and $y$. Then, their neighbourhood in the original graph $G$ differs only by the vertex $w$ (thus there can only be pairs of twins, e.g. no triplets). Furthermore, $S \sm \{x,y\}$ is 2-incident and $G' \star^2 S = (G' \star^2 S \sm \{x,y\}) \star^1 \{x\}$. Let $S'$ be the set obtained from $S$ by removing every pair of twins and every vertex of degree 1. $S'$ is 2-incident and $|S'|\ls 10$. A non-empty 2-incident independent set without twins or vertex of degree 1 is genuine, and thus by \cref{lemma:exp_support} has at least 11 vertices. Thus, $S' = \emptyset$.

    $w$ was chosen arbitrarily in $w \in \delta(S) \sm \{u,v\}$. Thus, for any $w\in \delta(S) \sm \{u,v\}$, the set obtained from $S$ by removing every pair of twins and every vertex of degree 1 in the graph $G[V \sm \{w\}]$, is empty. 

    Let us prove that $S$ contains a vertex $x$ such that $\delta(S) \sm \{u,v\} \se N_G(x)$. Suppose by contradiction that this is not the case and consider a vertex $y \in S$ such that $|N_G(y) \cap (\delta(S) \sm \{u,v\})|$ is maximum. Note that $|N_G(y) \cap (\delta(S) \sm \{u,v\})| \gs 2$, else  $G \star^2 S = G$, as $S$ is 2-incident. By hypothesis there is a vertex $w \in \delta(S) \sm \{u,v\}$ such that $w \notin N_G(y)$. In $G[V \sm \{w\}]$, $y$ is not of degree 1 so it has a twin $z$, implying that $N_G(z) = N_G(y) \cup \{w\}$. This is a contradiction with the fact that $|N_G(y) \cap (\delta(S) \sm \{u,v\})|$ is maximum.

    As $|\delta_G(S)| \gs 7$, $|N_G(x) \cap (\delta(S) \sm \{u,v\})| \gs 5$. For any $w\in \delta(S) \sm \{u,v\}$, in $G[V \sm \{w\}]$ each vertex of degree at least 2 has a twin. Thus, in $G$, for any $w\in \delta(S) \sm \{u,v\}$, there is a vertex $x_w \in S$ such that $N_G(x_w) = N_G(x) \sm \{w\}$. Additionally, for any such vertex $x_w$, for any $w'\in \delta(S) \sm \{u,v,w\}$, there is a vertex $x_{w,w'} \in S$ such that $N_G(x_{w,w'}) = N_G(x_w) \sm \{w\}$. Thus, $S$ contains at least $1+5+\binom{5}{2} = 16$ vertices, contradicting the hypothesis that $|S|\ls 12$.
\end{proof}

\end{document}